\documentclass[11pt,letter]{article}
\usepackage{fullpage}
\usepackage{times}

\usepackage{amsmath,amsthm,amssymb,hyperref,verbatim}

\usepackage{natbib}
\allowdisplaybreaks
\usepackage{hyperref}
\hypersetup{
     colorlinks   = true ,% Colours links instead of boxes
     urlcolor     = blue,
     linkcolor    = blue,
     citecolor    = blue
}

\def\drawfigures{1}

\newtheorem{theorem}{Theorem}
\newtheorem{lemma}[theorem]{Lemma}
\newtheorem{corollary}[theorem]{Corollary}
\newtheorem{definition}[theorem]{Definition}
\newtheorem{claim}[theorem]{Claim}

\DeclareMathOperator{\E}{\mathbb{E}}
\DeclareMathOperator*{\R}{\mathbb{R}}

\DeclareMathOperator*{\one}{\mathbb{I}}
\DeclareMathOperator*{\argmax}{arg\,max}

\newcommand{\binarydist}{\mathcal{F}^{0/1}}

\newcommand{\mechanism}{\mathcal{M}}

\newcommand{\oracle}{\mathcal{V}}
\newcommand{\cohort}{\mathcal{C}}

\newcommand{\val}{\textup{val}}
\newcommand{\dist}{\textup{dist}}

\newcommand{\adist}{\textup{avdist}}
\newcommand{\pos}{\textup{pos}}
\newcommand{\score}{\textup{sc}}
\newcommand{\pl}{\textup{plu}}

\newcommand{\implied}{\textup{im}}
\newcommand{\art}{\textup{art}}
\newcommand{\impliedsw}{\sw_{\implied}}
\newcommand{\artsw}{\sw_{\art}}

\newcommand{\med}{\textsc{Mean}}
\newcommand{\pick}{\textsc{RtMean}}
\newcommand{\rand}{\textsc{RtSearch}}

\DeclareMathOperator{\sw}{\textup{SW}}
\DeclareMathOperator{\F}{\mathcal{F}}
\DeclareMathOperator{\PP}{\mathcal{P}}
\DeclareMathOperator{\PPi}{\PP}
\DeclareMathOperator*{\boxv}{\textup{Box}}
\DeclareMathOperator*{\rangev}{\textup{Range}}
\DeclareMathOperator*{\balanced}{\textup{Balanced}}
\DeclareMathOperator*{\pbalanced}{\textup{PartiallyBalanced}}

\DeclareMathOperator*{\distinct}{\textup{Distinct}}

\def\DEBUG{true}

\ifdefined\DEBUG
\fi

\usepackage{tikz}
\usetikzlibrary{shapes.geometric,decorations.pathreplacing,calligraphy,patterns,patterns.meta}

\title{\bf Beyond the worst case:\\Distortion in impartial culture electorates}

\author{Ioannis Caragiannis\thanks{Department of Computer Science, Aarhus University, {\AA}bogade 34, 8200 Aarhus N, Denmark. Email: \{\href{mailto:iannis@cs.au.dk}{iannis}, \href{mailto:karl@cs.au.dk}{karl}\}@cs.au.dk.} \and Karl Fehrs\footnotemark[1]}
\date{}

\begin{document}

\maketitle

\begin{abstract}
{\em Distortion} is a well-established notion for quantifying the loss of social welfare that may occur in voting. As voting rules take as input only ordinal information, they are essentially forced to neglect the exact values the agents have for the alternatives. Thus, in worst-case electorates, voting rules may return low social welfare alternatives and have high distortion. Accompanying voting rules with a small number of cardinal queries per agent may reduce distortion considerably.

To explore distortion beyond worst-case conditions, we use a simple stochastic model according to which the values the agents have for the alternatives are drawn independently from a common probability distribution. This gives rise to so-called {\em impartial culture electorates}. We refine the definition of distortion so that it is suitable for this stochastic setting and show that, rather surprisingly, all voting rules have high distortion {\em on average}. On the positive side, for the fundamental case where the agents have random {\em binary} values for the alternatives, we present a mechanism that achieves approximately optimal average distortion by making a {\em single} cardinal query per agent. This enables us to obtain slightly suboptimal average distortion bounds for general distributions using a simple randomized mechanism that makes one query per agent. We complement these results by presenting new tradeoffs between the distortion and the number of queries per agent in the traditional worst-case setting.
\end{abstract}

% Paper body
\section{Introduction}
Voting has been the subject of social choice theory for centuries. Traditionally, having elections as the main application area in mind, social choice theorists have made a lot of progress in understanding the axiomatic properties of {\em voting rules} (e.g., see~\citet{Z16} for an introduction to basic voting axioms). However, the use of voting goes far beyond elections, and today, it provides a compelling way of collective decision-making. So, in addition to the traditional axiomatic treatment of voting rules, other approaches that have an optimization flavour have attracted a lot of attention, starting with the work of Young (see~\citet{Y95} and references therein).

Among others, the utilitarian approach~\citep{BCHLPS15} assumes that voters have values for the alternatives. The preferences expressed in the ballot of a voter are a very short summary of these values, e.g., a ranking of the candidates in terms of their values, a set of candidates with values exceeding a threshold, or just the alternative with the highest value for the voter. A voting rule takes as input the voters' ballots and computes an outcome (e.g., a winning alternative). As it has no access to the underlying values of the voters for the alternatives, the outcome depends only on the short summaries in the voters' ballots. Rather unsurprisingly, the outcome of the voting rule will be sub-optimal if evaluated in terms of the voters' values for the alternatives.

How far can the outcome be from the optimal alternative? To answer this question, we need to specify a measure to assess the quality of alternatives. Following a rather standard approach in the literature, we use the {\em social welfare} ---the total value the agents have for an alternative--- as our quality measure. Then, the following question arises naturally. How far from the optimum can the outcome of a voting rule be in terms of social welfare? The notion of {\em distortion}, introduced by \citet{PR06}, comes to answer this question. The distortion of a voting rule is defined as the worst-case ratio, among all voting profiles on a set of alternatives, between the maximum social welfare among all alternatives and the social welfare of the winning alternative returned by the voting rule.

The distortion of voting rules has been the subject of a long list of papers in computational social choice for more than ten years now; see the nice survey of the key results in the area by~\citet{AFSV21}. For example, under mild assumptions about the valuations, the ubiquitous plurality rule has a distortion of $O(m^2)$, where $m$ is the number of alternatives~\citep{CP11}. Nevertheless, the distortion is inherently high. Even when the values for the alternatives are restricted (e.g., the total value of each agent for all alternatives is normalized to $1$), the distortion of any (possibly randomized) voting rule can be as bad as $\Theta(\sqrt{m})$~\citep{BCHLPS15,EKPS22}.

A recent approach by~\citet{ABFV21} aims to bypass such lower bounds by making very limited use of the underlying cardinal information. Here, besides the ranking submitted as a ballot by an agent, the voting rule (or, better, the {\em mechanism}) can pose queries to the agent regarding her value for particular alternatives. Even though optimal distortion results are now possible (by naively querying the values of all alternatives in each agent's ranking), the important problem to be solved is to design mechanisms that achieve low distortion by making a limited number of queries per agent. Among other results, \citet{ABFV21} present an algorithm that achieves constant distortion by making at most $O(\log^2{m})$ queries per agent. On the negative side, they show that any mechanism that achieves constant distortion must make at least $\Omega\left(\frac{\log{m}}{\log\log{m}}\right)$ per agent. Both results refer to deterministic mechanisms.

By the definition of distortion, the results discussed above are of a worst-case flavour. Voting rules and mechanisms are evaluated on a profile for which they have the worst possible performance, no matter how frequently such a profile may appear in practice. As such, they may not be suitable to explain the success of certain voting rules in practice.

\subsection{Overview of our contribution}
We follow the utilitarian framework but ---motivated by the recent trend of analyzing algorithms from non-worst-case perspectives~\citep{R20}--- attempt an evaluation of voting rules on an {\em average-case} basis. At the conceptual level, we introduce the notion of {\em average distortion} for such stochastic preferences. To this end, we consider a simple stochastic setting in which each agent draws a random value for each alternative according to a common probability distribution. The draws of each agent for each alternative are independent. Naturally, {\em impartial culture electorates}~\citep{PW09,TRG03} ---i.e., profiles with uniformly random rankings of alternatives--- emerge in this way. The average distortion of a mechanism is defined as the expected maximum social welfare among the alternatives over the expected social welfare of the alternative returned by the mechanism. To distinguish between the two notions, we use the term {\em worst-case distortion} to refer to the traditional definition. It is not hard to see that the average distortion is always upper-bounded by the worst-case distortion.%\pagebreak

We warm up by showing (in Section~\ref{sec:omega_m_lower_bound_exp_distortion}) that, perhaps surprisingly, any (potentially randomized) voting rule has average distortion $\Omega(m)$, implying that the trivial voting rule that selects one of the alternatives uniformly at random (ignoring the profile) has an almost best possible average distortion. Our lower bound uses a very simple {\em binary} probability distribution (i.e., one that returns values of $1$ and $0$). In light of this seemingly discouraging result, our primary objective is to understand the additional power a limited number of value queries can give to a voting rule by seeking answers to the following questions:

\begin{quote}
Can we design mechanisms that use a few (ideally, a constant number of) queries per agent and have low (ideally, constant) average distortion? What conditions do we need to impose on the distribution from which the agents' valuations are drawn in order to achieve this?
\end{quote}

Notice that we do not intend to introduce other restrictions on the electorates, such as requiring the number of agents $n$ to be large compared to the number of alternatives $m$. Instead, we would like our mechanisms to work for {\em any} choice of $n$ and $m$.

Our first positive result is for the family of binary distributions. In Section~\ref{sec:1-query}, we present a deterministic mechanism, called \med, which makes a {\em single query} per agent and achieves {\em constant} average distortion. This is our most technically involved result and indicates that using queries (even in a minimal way) can lead to a significant improvement in average distortion.

Binary distributions are, of course, an exceedingly crude model for the agents' valuations.  However, we demonstrate the usefulness of our result by employing the mechanism \med\ as a building block for a randomized mechanism, which works for a general distribution $F$. This mechanism, which we call \pick, still requires only one query per agent and achieves an average distortion of $O\left(\log{m}+\log{\frac{\sigma^2}{\mu^2}}\right)$, where $\mu$ and $\sigma^2$ are the mean and variance of the distribution $F$. For many distributions of interest, the average distortion bound obtained is only logarithmic in $m$. A similar idea is used to define our randomized mechanism \rand\ for the traditional model of worst-case distortion. \rand\ uses a logarithmic number of queries per agent and achieves logarithmic worst-case distortion. Such an upper bound is not known for deterministic mechanisms. These results are presented in Section~\ref{sec:general-upper-bounds}. To the best of our knowledge, this is the first analysis of randomized mechanisms that make value queries in the distortion literature.

Our upper bound on the average distortion of mechanism \med\ is not attainable by deterministic mechanisms in the traditional worst-case model. We prove that no deterministic mechanism that makes a single query per agent can achieve distortion better than $\Omega(\sqrt{m})$, even when the valuations are binary. More importantly, for general valuations, we present a new lower bound of $\Omega(\log{m})$ on the number of queries per agent that allows for constant worst-case distortion. This improves the previously best-known lower bound of \citet{ABFV21} by a sublogarithmic factor. These results appear in Section~\ref{sec:lower-bounds}.

\subsection{Further related work}
Using different methodological approaches, utilities in voting have been considered in social choice theory since the work of~\citet{B1780} in the 18th century; for recent work on the topic see~\citet{ABF11,Pivato16}. The stochastic model we follow was introduced by~\citet{BCHLPS15}. Among other investigations, \citeauthor{BCHLPS15} aimed at identifying the optimal voting rule in terms of the expected social welfare of the winning alternative. Their main conclusion is that this voting rule is a positional scoring rule with parameters depending on the probability distribution the agents' values are drawn from. They do not consider notions related to average distortion and do not present related bounds. \citet{GKPSZ23} define a version of distortion (slightly different from ours) in the model of~\citeauthor{BCHLPS15} as follows. For a particular profile, they examine the ratio between the expected maximum social welfare and the expected social welfare of the winning alternative, both terms conditioned on the random values of the alternatives being consistent with the profile. They define the expected distortion as the worst-case ratio over all profiles. It can be verified that, for every voting rule, the value of expected distortion according to~\citet{GKPSZ23} is bounded above and below by the worst-case distortion and our average distortion, respectively. The main result of~\citet{GKPSZ23} is that binomial voting (which, in the terminology of \citet{BCHLPS15}, is essentially the average-case optimal voting rule for a fair Bernoulli trial) approximates the rule with minimum expected distortion for values drawn from distributions supported on $[0,1]$. Both~\citet{BCHLPS15} and~\citet{GKPSZ23} restrict their attention to voting rules and do not consider value queries to the agents.

In the current paper, we assume that agents have non-negative values for the alternatives. In a series of recent papers on {\em metric voting}, originating from the work of~\citet{ABP15}, agents and alternatives are assumed to be located in a metric space, and the preferences of each agent reflect her relative distance from the alternatives. In that different setting, the distortion quantifies the suboptimality of the outcome of voting rules in terms of the {\em social cost}. This line of research has led to challenging algorithmic problems with beautiful solutions~\citep{GHS20,KK22,CR22,CRWW24}. \citet{CDK17,CDK18} and~\citet{GLS19} consider variants of average-case distortion that, due to the metric setting and additional modelling assumptions, are only distantly related to ours.

The study of the utilitarian framework in the recent CS literature goes beyond the study of single-winner voting. \citet{CNPS17,CSV22} and \citet{BNPS21} consider multiwinner voting and participatory budgeting settings, respectively. \citet{FMV20} investigate voting in distributed settings.
The utilitarian approach has also been used in settings that are more general than voting. The main idea is to explore how well algorithms that use only ordinal information about the underlying input (as opposed to cardinal values) can approximate the optimal solutions of combinatorial optimization problems. The survey of~\citet{AFSV21} covers early work on this hot topic, as well as on the other directions discussed above. The benefits of allowing for limited access to the cardinal information were more recently explored for matching~\citep{MML21,LV24,ABFV22a,ABFV22b} and clustering~\citep{BCF+} problems.

We remark that stochastic models are otherwise ubiquitous in the EconCS literature, including a long line of research in mechanism design and auctions originating from the seminal paper of~\citet{M81}, and a rich menu of other settings like matchings~\citep{CIK+09}, fair division~\citep{MS20}, kidney-exchange~\citep{TP15}, and many more.

\section{Preliminaries}\label{sec:prelim}

Throughout the paper, we denote by $N$ and $A$ the sets of agents and alternatives and reserve $n$ and $m$ for their cardinalities, respectively. Each agent $i\in N$ has a {\em ranking} $\succ_i$ of the alternatives, i.e., a strict ordering of the elements in $A$. A {\em profile} $P=\{\succ_i\}_{i\in N}$ is just a collection of the agents' rankings. We denote by $\PP$ the set of all possible profiles with $n$ agents and $m$ alternatives. A voting rule $\mechanism:\PP\rightarrow A$ takes as input a profile and returns a single winning alternative.

We assume that the ranking of each agent results from underlying hidden non-negative {\em values} that the agent has for each alternative. For $i\in N$, the {\em valuation function} $\val_i:A \rightarrow \R_{\geq 0}$ returns the values of agent $i$ for the alternatives in $A$. Then, agent $i$'s ranking $\succ_i$ is {\em consistent} with $\val_i$. This means that $a \succ_i a'$ only if $\val_i(a) \geq \val_i(a')$ for every pair of alternatives $a,a' \in A$. Given a ranking $\succ$, the function $\pos_{\succ}:A \rightarrow [m]$ returns the position of a given alternative in the ranking. We say that alternative $a$ is the {\em top-ranked} alternative of agent $i$ if $\pos_{\succ_i}(a) = 1$.\pagebreak

Consider a set of valuation functions $v = \{\val_i\}_{i\in N}$. We typically refer to the set $v$ as the agents' {\em valuations}. The {\em social welfare} of an alternative $a \in A$ is defined as
$$\sw(a,v) = \sum_{i \in N} \val_i(a).$$

Let $\oracle$ be the set of all possible valuations the agents in $N$ can have for the alternatives in $A$. We denote by $\PP(\val_i)$ the set of rankings that are consistent with the valuation function $\val_i$ of agent $i$. We say that a profile $P=\{\succ_i\}_{i\in N}$ is consistent with the valuations $v$ if the ranking $\succ_i$ of every agent $i\in N$ is consistent with her valuation function $\val_i$, i.e., $\succ_i\in \PP(\val_i)$. Then, $\PP(v)$ represents the set of profiles in $\PP$ which are consistent with the valuations $v$.

Besides voting rules, we consider {\em mechanisms} that have, in addition to the profile, access to the valuations. Such a mechanism 
$\mechanism:\PP\times \oracle \rightarrow A$ takes as input the profile and the valuations and returns a winning alternative. We are particularly interested in mechanisms that use the whole profile on input but only a small part of the valuations by making a limited number of {\em queries} per agent. A query for the value of agent $i\in N$ for alternative $a\in A$ simply returns the value of $\val_i(a)$.

The {\em worst-case distortion} of a mechanism $\mechanism$ applied on profiles consistent with valuations from $\oracle$ is defined as
$$\dist(\mechanism) = \sup_{\substack{v\in \oracle\\P\in \PP(v)}}\frac{\max_{a \in A} \sw(a,v)}{\sw(\mechanism(P,v), v)},$$
i.e., it is the worst-case ratio---among all valuations and consistent profiles---between the maximum social welfare and the social welfare of the alternative returned by the mechanism. For randomized mechanisms, in which the alternative returned is a random variable, 
we use the expectation $\E[\sw(\mechanism(P,v), v)]$ (taken over the random choices of $\mechanism$) instead of $\sw(\mechanism(P,v), v)$ in the denominator.

We extend the notion of distortion to {\em stochastic environments} with random valuations and consistent profiles. We assume that the values of the agents for the alternatives are drawn from a {\em known} common distribution $F$. In particular, for each agent $i\in N$ and alternative $a\in A$, the value $\val_i(a)$ is drawn independently from distribution $F$. Given valuations $v$ selected in this way, a consistent profile $P$ is selected uniformly at random among all profiles in $\PP(v)$ (essentially, in her ranking, each agent breaks ties among alternatives in terms of value uniformly at random). This gives rise to uniformly random profiles, which are known as {\em impartial culture electorates} in the social choice theory literature.

For valuations $v$, we use $v \sim F$ to denote that, for every alternative $a\in A$ and every agent $i\in N$, the value $\val_i(a)$ is drawn independently from $F$. We use the notation $\succ_i\sim \PP(\val_i)$ to refer to a ranking that is selected uniformly at random among all rankings that are consistent with the valuation function $\val_i$ of agent $i$. Similarly, we use $P\sim \PP(v)$ for a profile that is selected uniformly at random among all profiles that are consistent with valuations $v$. Then, the {\em expected social welfare} of the winning alternative picked by a deterministic mechanism $\mechanism$ is
$$\E_{\substack{v\sim F\\ P \sim \PP(v)}}\left[\sw(\mechanism(P,v), v)\right].$$
When mechanism $\mechanism$ is randomized, the expectation is taken over the randomness of $\mechanism$ as well.

We now adapt the notion of distortion to this stochastic setting. The {\em average distortion} of a mechanism $\mechanism$ on a family of distributions $\F$ is
$$\adist(\mechanism,\F)
= \sup_{F\in \F}
\frac{\E_{v\sim F}\left[\max\limits_{a \in A} \sw(a,v)\right]}
{\E_{\substack{v\sim F\\ P \sim \PP(v)}}\Big[\sw(\mechanism(P,v), v)\Big]}\,.$$
We remark that the trivial voting rules that return a fixed alternative or an alternative selected uniformly at random have average distortion at most $m$ for every distribution $F$. Denoting them by $\mechanism$, this is due to the simple fact that
$$\E_{\substack{v\sim F\\ P \sim \PP(v)}}\Big[\sw(\mechanism(P), v)\Big]=\frac{1}{m}\cdot \E_{v\sim F}\left[\sum_{a \in A}{\sw(a,v)}\right]\geq \frac{1}{m}\cdot \E_{v\sim F}\left[\max_{a \in A}{\sw(a,v)}\right].$$
For particular distributions, the average distortion can be much lower. As an example, consider the uniform distribution over the interval $[a,b]$ with $a,b\geq 0$. Then, the expected social welfare of the alternative returned by both trivial rules is $n\cdot (a+b)/2$ while the expected maximum social welfare (over all alternatives) cannot exceed $n\cdot b$. Hence, the average distortion is at most $2$ in this case.

A fundamental family that is of central importance to our work is the family $\binarydist$ of {\em binary distributions}, consisting of the set of probability distributions $\{F_p\}_{p\in [0,1]}$, where $F_p$ is such that the random variable $z$ that is drawn according to $F$ is equal to $1$ with probability $p$ and to $0$ with probability $1-p$.

\section{Average distortion can be high}\label{sec:omega_m_lower_bound_exp_distortion}

We begin our technical exposition with a lower bound of $\Omega(m)$ on the average distortion. Essentially, Theorem~\ref{thm:omega_m_exp_distortion} implies that the trivial voting rules that return a fixed alternative or an alternative selected uniformly at random on every profile have approximately optimal average distortion. 

In our proof, we make use of a particular distribution $F$ from the family $\binarydist$. Furthermore, we exploit a class of voting rules called {\em positional scoring rules}. A positional scoring rule $g$ is defined by a scoring vector $\langle \alpha_1, \alpha_2, ..., \alpha_m\rangle$ with $\alpha_1\geq \alpha_2\geq ...\geq a_m\geq 0$. On input a profile $P$, the voting rule $g$ assigns a score $\score_g(a,P)$ to every alternative $a\in A$ and the winning alternative is the one with the highest score (breaking ties according to some tie-breaking rule). The score of an alternative is computed as follows. For $j=1, ..., m$, alternative $a$ takes $\alpha_j$ points each time it appears in the $j$-th position in an agent's ranking, i.e., $\score_g(a,P)=\sum_{i\in N}{\alpha_{\pos_{\succ_i}(a)}}$. For example, {\em plurality} is the positional 
scoring rule that uses the $m$-entry scoring vector $\langle 1, 0, ..., 0\rangle$. The plurality winner on profile $P$, denoted by $\pl(P)$, is the alternative that appears most often in the top position of the agents' rankings.

\begin{theorem}\label{thm:omega_m_exp_distortion}
For every (possibly randomized) mechanism $\mechanism$, $\adist(\mechanism,\binarydist)\in \Omega(m)$.
\end{theorem}

\begin{proof}
Consider impartial culture electorates with $m$ alternatives, $n\geq 6m\ln{m}$ agents, and binary values drawn from the probability distribution $F\in \binarydist$ with $p=\frac{1}{nm}$. Then, the probability that some alternative has positive social welfare is $1-\left(1-\frac{1}{nm}\right)^{nm}\geq 1-e^{-1}$. Thus, $\E_{v\sim F}[\max_{a\in A}{\sw(a,v)}]\geq 1-e^{-1}$.
Now consider any voting rule $\mechanism$; we will complete the proof by showing that 
\begin{align}\label{eq:4-over-m}
\E_{\substack{v\sim F\\P\sim \PP(v)}}[\sw(\mechanism(P),v)] &\leq \frac{4}{m}.
\end{align}
Below, we prove inequality~(\ref{eq:4-over-m}), assuming that $\mechanism$ is deterministic. For randomized $\mechanism$, the expectation in the LHS of~(\ref{eq:4-over-m}) should also be taken over the randomness of $\mechanism$ as well. In that case, inequality (\ref{eq:4-over-m}) follows trivially by our arguments, interpreting $\mechanism$ as a probability distribution over deterministic voting rules.

Let $g$ be the positional scoring rule that uses the scoring vector $\langle \alpha_1, \alpha_2, ..., \alpha_m\rangle$ with 
$$\alpha_j=\E_{\substack{\val_i\sim F\\\succ_i\sim \PPi(\val_i)}}[\val_i(a)|\pos_{\succ_i}(a)=j],$$ i.e., score $\alpha_j$ is equal to the expected value according to $F$ given by each agent to the alternative ranked at position $j$.\footnote{\citet{BCHLPS15} refer to this positional scoring rule $g$ as an average-case optimal social choice function.} Now, observe that for a given profile $P=\{\succ_i\}_{i\in N}$ and alternative $a\in A$, we have 
\begin{align*}
\E_{v\sim F}[\sw(a,v)|P\in \PP(v)]
&=\sum_{i\in N}{\E_{\val_i\sim F}[\val_i(a)|\succ_i \in \PP(\val_i)]}
=\sum_{i\in N}{\alpha_{\pos_{\succ_i}(a)}}=\score_g(a,P),
\end{align*}
where $\score_g(a,P)$ denotes the score of alternative $a$ in profile $P$ according to positional scoring rule $g$. Thus,
\begin{align}\nonumber
\E_{\substack{v\sim F\\P\sim \PP(v)}}[\sw(\mechanism(P),v)] &= \sum_{P\in \PP}{\E_{v\sim F}[\sw(\mechanism(P),v)|P\in \PP(v)]\cdot \Pr_{v\sim F}[P\in \PP(v)]}\\\nonumber
&= \sum_{P\in \PP}{\score_g(\mechanism(P),P) \cdot \Pr_{v\sim F}[P\in \PP(v)]}\\\label{eq:sw-vs-score-g}
&= \E_{\substack{v\sim F\\P\sim \PP(v)}}[\score_g(\mechanism(P),P)].
\end{align}
Notice that, by the definitions of voting rules \pl\ and $g$, for any profile $P$ and alternative $a\in A$, it holds
\begin{align}\label{eq:score-g-vs-plurality-score}
\score_g(a,P) &\leq \alpha_1 \cdot \score_{\pl}(a,P)+n\cdot \alpha_2 \leq \alpha_1 \cdot \score_{\pl}(\pl(P),P)+n\cdot (mp-\alpha_1).
\end{align}
The first inequality follows since $g$ gives $\alpha_1$ points to alternative $a$ for each of its $\score_\pl(a,P)$ appearances in the top position in the agents' rankings and at most $\alpha_2$ points for its remaining appearances. The second inequality follows since the plurality winner $\pl(P)$ has the highest plurality score (and, hence, $\score_\pl(a,P)\leq \score_\pl(\pl(P),P)$ for every alternative $a\in A$) and since $\alpha_2\leq \sum_{j=1}^m{\alpha_j}-\alpha_1$ and $\sum_{j=1}^m{\alpha_j}=mp$ (as the expected total value given to all alternatives by an agent).

We will also need the following claim.
\begin{claim}\label{claim:plu-score}
$\E_{\substack{v\sim F\\P\sim \PP(v)}}[\score_{\pl}(\pl(P),P)] \leq \frac{3n}{m}$.
\end{claim}

\begin{proof}
The proof will follow by a simple application of the Chernoff bound, e.g., see~\citet{motwani95}.
\begin{lemma}[Chernoff bound, upper tail]\label{lem:chernoff}
For every binomial random variable $Q$ and any $\delta>0$, we have
\begin{align*}
\Pr[Q\geq (1+\delta)\E[Q]] \leq \exp\left(-\frac{\delta^2\E[Q]}{2+\delta}\right).
\end{align*}
\end{lemma}

Consider an alternative $a\in A$ and denote by $X_i$ the random variable indicating whether agent $i$ ranks $a$ first ($X_i=1$; this happens with probability $1/m$) or not ($X_i=0$). Observe that $\score_\pl(a,P)=\sum_{i\in N}{X_i}$, i.e., $\score_\pl(a,P)$ is a binomial random variable with expectation $\frac{n}{m}$. By applying Lemma~\ref{lem:chernoff} with $\delta=1$, we obtain that $\Pr[\score_\pl(a,P)\geq \frac{2n}{m}]\leq \exp\left(-\frac{n}{3m}\right)\leq \frac{1}{m^2}$. The last inequality follows since $n\geq 6m\ln{m}$. Taking the union bound over the $m$ alternatives then gives $\Pr[\score_\pl(\pl(P),P)\geq \frac{2n}{m}]\leq\frac{1}{m}$. Finally, since the plurality score never exceeds $n$, we have $\E_{\substack{v\sim F\\P\sim \PP(v)}}[\score_\pl(\pl(P),P)]\leq \frac{2n}{m}\cdot \left(1-\frac{1}{m}\right)+n\cdot \frac{1}{m}\leq \frac{3n}{m}$.
\end{proof}

Notice that the top alternative of an agent has value $0$ with probability $(1-p)^m$ and value $1$ otherwise. Thus, $\alpha_1=1-(1-p)^m$. Now, observe that $1-(1-p)^m=1-\left(1-\frac{1}{nm}\right)^m\geq 1-\frac{1}{e^{1/n}}\geq \frac{1}{n+1}$, using the properties $(1-r/t)^t\leq e^{-r}$ for $t>0$ and $r\geq 0$ and $e^r\geq 1+r$. Thus, $mp-\alpha_1\leq \frac{1}{n}-\frac{1}{n+1}<\frac{1}{n^2}$. Also, using the property $(1-r)^t\geq 1-rt$ for $t\geq 1$, we get $\alpha_1=1-(1-p)^m\leq pm=\frac{1}{n}$. Combining these two last inequalities with equations (\ref{eq:sw-vs-score-g}) and (\ref{eq:score-g-vs-plurality-score}) and Claim~\ref{claim:plu-score}, we obtain
\begin{align*}
\E_{\substack{v\sim F\\P\sim \PP(v)}}[\sw(\mechanism(P),v)]
&= \E_{\substack{v\sim F\\P\sim \PP(v)}}[\score_g(\mechanism(P),P)]\\
&\leq \alpha_1\cdot \E_{\substack{v\sim F\\P\sim \PP(v)}}[\score_\pl(\pl(P),P)] + n(mp-\alpha_1)
\leq \frac{3}{m}+\frac{1}{n}\leq \frac{4}{m},
\end{align*}
as desired by (\ref{eq:4-over-m}). The last inequality follows from the relation between $n$ and $m$.
\end{proof}

\section{Constant average distortion with a single query per agent}\label{sec:1-query}

In the previous section, we saw that the average distortion can be high even for binary distributions if the mechanism has no access to the agents' underlying valuations. Still, binary valuations may reveal a lot of information with a single query. Consider a query for the value that an agent has for the alternative at position $j$ of her ranking. If the query returns a value of $1$, this means that the agent has a value of $1$ for all alternatives in positions higher than $j$ as well. This observation motivates the following definition of {\em implied social welfare}.

\begin{definition}[implied social welfare]\label{defn:implied_sw}
For a distribution $F \in \binarydist$, assume that a mechanism queried each agent for the value of the alternative in position $k$ of her ranking. For each agent $i\in N$, let $\val_{i,k}$ denote the respective value. The {\em implied social welfare} then is
$$\impliedsw(a,P,v,k) = \sum_{i \in N} \one\{\val_{i,k} = 1 \,\wedge\, \pos_{\succ_i}(a) \leq k\}.$$
\end{definition}

\noindent We use the notion of implied social welfare in the definition of the following mechanism.

\begin{definition}[mechanism \med]\label{defn:median_mechanism}
Given a profile $P$ with underlying valuations drawn from a distribution $F_p\in\binarydist$, the mechanism \med\ queries each agent for the value of the alternative at position $\tau = \max\{1, \lfloor pm\rfloor\}$. The mechanism then returns the alternative that maximizes the implied social welfare, that is,
$$\med(P,v) \in \argmax_{a \in A} \impliedsw(a,P,v,\tau),$$
breaking ties arbitrarily.
\end{definition}

We show that mechanism \med\ has constant average distortion with a single query for the family of binary distributions.\footnote{One may wonder whether the simplification of mechanism \med, which always queries the value of the top-ranked alternative in each agent and returns the alternative that maximizes the implied social welfare, achieves a constant average distortion as well. We discuss this question in Appendix~\ref{app:sec:comments} and show that, unless we introduce a more sophisticated tie-breaking, this variation has average distortion at least $\Omega\left(\frac{\log{m}}{\log\log{m}}\right)$.} Notably, it is impossible for deterministic 1-query mechanisms to achieve similar guarantees in the traditional setting of worst-case distortion. Indeed, even for binary valuations, the worst case distortion of these mechanisms is $\Omega(\sqrt{m})$; see Theorem~\ref{thm:0/1_valued_1_query_lower_bound}.

\begin{theorem}\label{thm:mean_mechanism_distortion}
Mechanism $\med$ has average distortion at most $27$ in impartial culture electorates with $n$ agents and $m$ alternatives, and underlying values drawn from any probability distribution in $\binarydist$.
\end{theorem}

\begin{proof}
We partition family $\binarydist$ into the subfamilies of probability distributions $\F_1$, $\F_2$, and $\F_3$ defined as follows:
\begin{align*}
\F_1 &= \{F_p\in \binarydist : p\geq 1/m\}\\
\F_2 &= \{F_p\in \binarydist : 1-(1-1/n)^{1/m}\leq p < 1/m \}\\
\F_3 &= \{F_p\in \binarydist : p < 1-(1-1/n)^{1/m}\}
\end{align*}
Let $F_p\in \binarydist$; we will prove the theorem by distinguishing between the three cases $F_p\in \F_1$, $F_p\in \F_2$, and $F_p\in \F_3$. 

We introduce some notation that we use throughout this proof. Let $\tau$ denote the position that mechanism $\med$ queried in each agent's ranking, i.e., $\tau=\max\{1,\lfloor mp\rfloor\}$. For valuations $v$,  let $X_i(v)$ be the number of alternatives for which agent $i$ draws a value of $1$ and let $X(v)=\sum_{i\in N}{X_i(v)}$. We will consider valuations $v$ drawn from the probability distribution $F_p$. Thus, $X_i(v)$ is a random variable following the binomial distribution with $m$ trials and success probability $p$. Since $X(v)$ is the sum over i.i.d.~binomially distributed random variables, $X(v)$ itself follows a binomial distribution with $nm$ trials and success probability $p$.

Furthermore, we denote by $Z_i(v,\tau)$ the random variable indicating whether the query at position $\tau$ of agent $i$'s ranking returned a value of $1$ (then, $Z_i(v,\tau)=1$) or not (then, $Z_i(v,\tau)=0$). Clearly, $Z_i(v,\tau)=\one\{X_i(v)\geq \tau\}$ such that the random variable $Z(v,\tau)=\sum_{i\in N}{Z_i(v,\tau)}$ follows a binomial distribution with $n$ trials and success probability $\Pr[X_i(v)\geq\tau]$. We note that the variables $X_i$ are identically distributed for every $i\in N$.

For the first two cases below, we will need a technical lemma which we prove in Appendix~\ref{app:sec:lem:balanced}.

\begin{lemma}\label{lem:balanced}
For integer $s\in [nm]$, define the condition $\balanced(v,s)$ to be $X(v)=s$ with $\lfloor s/n\rfloor \leq X_i(v) \leq \lceil s/n\rceil$ for $i\in N$. Then, for any distribution $F\in\binarydist$, it holds that
\begin{align*}
\E_{v\sim F}\left[\max_{a\in A}{\sw(a,v)}|X(v)=s\right] &\leq \E_{v\sim F}\left[\max_{a\in A}{\sw(a,v)}|\balanced(v,s)\right].
\end{align*}
\end{lemma}

\paragraph{Case 1.} Consider impartial culture electorates with $n$ agents and $m$ alternatives with underlying values drawn from the distribution $F_p\in \F_1$.  Denote by $\boxv(v)$ the condition $X_i(v)=\tau$ for $i\in S$ and $X_i(v)=0$ for $i\in N\setminus S$, where $S$ is a subset of the agents of size exactly $\lfloor n/2 \rfloor$. Now, define the quantity
\begin{align}\label{eq:benchmark-case-1}
B&=\E_{v\sim F}\left[\max_{a\in A}{\sw(a,v)}|\boxv(v)\right],
\end{align}
i.e., the expected maximum social welfare among all alternatives given that $\lfloor n/2\rfloor$ agents draw a value of $1$ for exactly $\tau$ alternatives and the remaining agents draw values of $0$ for all alternatives. The quantity $B$ will be a benchmark that will help us compare the expected social welfare of $\med(P,v)$ to the maximum social welfare.

As stated above, each $X_i(v)$ is a binomial random variable with $m$ trials and success probability $p$. Hence, its median is either at the value $\lfloor pm \rfloor$ or $\lceil pm \rceil$, which are both at least $\tau$ since $p\geq 1/m$. Thus, $\Pr_{v\sim F}[X_i(v) \geq \tau]\geq 1/2$. The random variable $Z(v,\tau)$ then follows a binomial distribution with $n$ trials and success probability at least $1/2$. By the same property of the median of the binomial distribution, it now holds that $\Pr_{v\in F}[Z(v,\tau)\geq \lfloor n/2\rfloor]\geq 1/2$. With this observation and by applying the law of total expectation, we have
\begin{align}\nonumber
&\E_{\substack{v\sim F_p\\P\sim \PP(v)}}[\sw(\med(P,v),v)]\\\nonumber
&= \E_{\substack{v\sim F_p\\P\sim \PP(v)}}[\sw(\med(P,v),v)|Z(v,\tau)\geq \lfloor n/2\rfloor]\cdot \Pr_{v\in F_p}[Z(v,\tau)\geq \lfloor n/2\rfloor]\\\nonumber
&\quad +\E_{\substack{v\sim F_p\\P\sim \PP(v)}}[\sw(\med(P,v),v)|Z(v,\tau)< \lfloor n/2\rfloor]\cdot \Pr_{v\in F_p}[Z(v,\tau)< \lfloor n/2\rfloor]\\\nonumber
&\geq \frac{1}{2}\cdot \E_{\substack{v\sim F_p\\P\sim \PP(v)}}\left[\sw(\med(P,v),v)|Z(v,\tau)\geq \lfloor n/2\rfloor\right]\\\nonumber
&\geq \frac{1}{2}\cdot \E_{\substack{v\sim F_p\\P\sim \PP(v)}}\left[\max_{a\in A}\impliedsw(a,P,v,\tau)|Z(v,\tau)\geq \lfloor n/2\rfloor\right]\\\label{eq:bound-for-mechanism-case-1}
&\geq \frac{1}{2}\cdot \E_{v\sim F_p}\left[\max_{a\in A}{\sw(a,v)}|\boxv(v)\right] = \frac{1}{2}\cdot B.
\end{align}
The first inequality above is obvious. By definition of $\med$, $\med(P,v)$ is the alternative that maximizes the implied social welfare when querying each agent in position $\tau$. Hence, the implied social welfare is a lower bound for the social welfare of $\med(P,v)$, which yields the second inequality. Finally, under the condition that $Z(v,\tau)\geq \lfloor n/2\rfloor$, there are at least $\lfloor n/2\rfloor$ agents $i$ for each of which $X_i(v) \geq \tau$. This includes the condition $\boxv(v)$ and the third inequality follows. The last equality follows from the definition of $B$ in (\ref{eq:benchmark-case-1}).

We proceed to bound the expected maximum social welfare from above in terms of $B$. For this purpose, we require another technical lemma which we prove in Appendix~\ref{app:sec:lem:j-bound-case-1}.

\begin{lemma}\label{lem:j-bound-case-1}
For every distribution $F\in\binarydist$ and any positive integer $j$, it holds that
\begin{align*}
\E_{v\sim F}\left[\max_{a\in A}{\sw(a,v)}|\balanced(v,jn\tau)\right] &\leq 3j\cdot B.
\end{align*}
\end{lemma}

For a positive integer $j$, define the condition $\rangev(v,j)$ to be $(j-1)n\tau < X(v)\leq jn\tau$. By applying the law of total expectation, we obtain
\begin{align}\nonumber
\E_{v\sim F_p}[\max_{a\in A}{\sw(a,v)}] &= \sum_{j=1}^{\infty}{\E_{v\sim F_p}[\max_{a\in A}{\sw(a,v)}|\rangev(v,j)]\cdot \Pr_{v\sim F_p}[\rangev(v,j)]}\\\nonumber
&\leq \sum_{j=1}^{\infty}{\E_{v\sim F_p}[\max_{a\in A}{\sw(a,v)}|X(v)=jn\tau]\cdot \Pr_{v\sim F_p}[\rangev(v,j)]}\\\nonumber
&\leq \sum_{j=1}^{\infty}{\E_{v\sim F_p}[\max_{a\in A}{\sw(a,v)}|\balanced(v,jn\tau)]\cdot \Pr_{v\sim F_p}[\rangev(v,j)]}\\ \label{eq:bound-for-mechanism-case-1-intermediate}
&\leq 3B\cdot \sum_{j=1}^{\infty}{j\cdot \Pr_{v\sim F_p}[\rangev(v,j)]}.
\end{align}
The first inequality follows since the condition $\rangev(v,j)$ includes the condition $X(v)=jn\tau$ and since the quantity $\E\left[\max_{a\in A}{\sw(a,v)}|X(v)=t\right]$ is non-decreasing in terms of $t$. The second inequality follows from Lemma~\ref{lem:balanced} while the third one follows from Lemma~\ref{lem:j-bound-case-1}. We now bound the term $\sum_{j=1}^{\infty}{j\cdot \Pr_{v\sim F_p}[\rangev(v,j)]}$.
\begin{align*}
\sum_{j=1}^{\infty}{j\cdot \Pr_{v\sim F_p}[\rangev(v,j)]}
&= \sum_{j=1}^{\infty} j \sum_{k=(j-1)n\tau}^{jn\tau}\Pr_{v\sim F_p}[X(v)=k]\\
&= \frac{1}{n\tau}\cdot \sum_{j=1}^{\infty} \sum_{k=(j-1)n\tau}^{jn\tau}jn\tau\cdot \Pr_{v\sim F_p}[X(v)=k]\\
&\leq \Pr_{v\sim F_p}[X(v)\leq n\tau] + \frac{2}{n\tau}\cdot \sum_{j=2}^{\infty} \sum_{k=(j-1)n\tau}^{jn\tau} k \cdot \Pr_{v\sim F_p}[X(v)=k]\\
&\leq \frac12 + \frac{2}{n\tau} \cdot \E_{v\sim F_p}[X(v)] \leq 4.5.
\end{align*}
The first inequality follows from the fact that $jn\tau/2 \leq (j-1)n\tau\leq k$ for $j\geq 2$. The second inequality follows from the definition of the expectation of random variable $X(v)$. Since $X(v)$ follows the binomial distribution with $nm$ trials and success probability $p$, we have $\E[X(v)]=nmp$ which is at most $2n\lfloor mp\rfloor=2n\tau$ since $p\geq 1/m$. 

Now, inequality~(\ref{eq:bound-for-mechanism-case-1-intermediate}) implies that $E_{v\sim F}[\max_{a\in A}{\sw(a,v)}] \leq 13.5B$. Combining this bound with inequality~(\ref{eq:bound-for-mechanism-case-1}) lets us conclude that $\adist(\med,\F_1)\leq 27$, as desired.

\paragraph{Case 2.} Consider impartial culture electorates with $n$ agents and $m$ alternatives with underlying values drawn from the distribution $F_p\in \F_2$. Since $p<1/m$, mechanism $\med$ always queries the value of the top-ranked alternative in each agent, i.e., $\tau = 1$. Then, $\med$ picks the alternative that maximizes the implied social welfare $\impliedsw(a,P,v,1)$. It is thus immediately clear that
\begin{align}\label{eqn:sw_mech_case_2}
\E_{\substack{v \sim F_p \\ P \sim \PP(v)}}[\sw(\med(P,v),v)] \geq \E_{\substack{v\sim F_p\\P\sim \PP(v)}}\left[\max_{a\in A} \impliedsw(a,P,v,1)\right].
\end{align}

We intend to also relate the expected maximum social welfare to the RHS of (\ref{eqn:sw_mech_case_2}). First, we observe that
\begin{align}\label{eqn:max_sw_case_2}
\E_{v\sim F_p}\left[\max_{a\in A} \sw(a,v)\right]
= \sum_{t=1}^n \E_{\substack{v\sim F_p\\P\sim \PP(v)}}\left[\max_{a\in A} \sw(a,v) | Z(v,1) = t\right] \cdot \Pr_{v\sim F_p}[Z(v,1) = t].
\end{align}
In the following, we will upper-bound the term $\E[\max_{a\in A} \sw(a,v) | Z(v,1) = t]$ and show that this quantity is within a constant factor of $\E[\max_{a\in A}\impliedsw(a,P,v,1) | Z(v,1) = t]$ for every positive $t$. We will need another technical lemma (see Appendix~\ref{app:sec:lem:t-bound-case-2} for the proof). Indeed, the proof of the Lemma~\ref{lem:t-bound-case-2} makes use of Lemma~\ref{lem:balanced} and can be seen as a refined formulation of Lemma~\ref{lem:j-bound-case-1} for the present case where $\tau = 1$.

\begin{lemma}\label{lem:t-bound-case-2}
For every distribution $F\in\binarydist$ and any positive integer $j$, it holds that
\begin{align*}
\E_{v\sim F_p}\left[\max_{a\in A} \sw(a,v)|Z(v,1)=t,X(v)=j\right]
&\leq \lceil j/t \rceil \cdot\E_{\substack{v\sim F_p\\P\sim \PP(v)}}\left[\max_{a\in A} \impliedsw(a,P,v,1)|Z(v,1) = t\right].
\end{align*}
\end{lemma}

With this lemma in hand, we have that
\begin{align}\nonumber
&\E_{v\sim F_p}\left[\max_{a\in A} \sw(a,v) | Z(v,1) = t\right]\\ \nonumber
&=\sum_{j=t}^{\infty} \E_{v\sim F_p}\left[\max_{a\in A} \sw(a,v)|Z(v,1)=t,X(v)=j\right]\cdot\Pr_{v\sim F_p}[X(v)=j|Z(v,1)=t]\\ \label{eqn:max_sw_case_2_cond_Z_t}
& \leq \E_{\substack{v\sim F_p\\P\sim \PP(v)}}\left[\max_{a\in A} \impliedsw(a,P,v,1)|Z(v,1) = t\right] \sum_{j=t}^{\infty}\left\lceil\frac{j}{t}\right\rceil \Pr_{v\sim F_p}[X(v)=j|Z(v,1)=t].
\end{align}
We proceed to bound the sum that appears in the previous inequality.
\begin{align}\nonumber
&\sum_{j=t}^{\infty}\left\lceil\frac{j}{t}\right\rceil \Pr_{v\sim F_p}[X(v)=j|Z(v,1)=t]\\ \nonumber
&\leq \sum_{j=t}^{\infty}\left(\frac{j}{t}+1\right) \Pr_{v\sim F_p}[X(v)=j|Z(v,1)=t]\\ \nonumber
&= \sum_{j=t}^{\infty} \Pr_{v\sim F_p}[X(v)=j|Z(v,1)=t] + \frac{1}{t} \sum_{j=t}^{\infty} j\cdot\Pr_{v\sim F_p}[X(v)=j|Z(v,1)=t]\\\nonumber
&=  1 + \frac{1}{t} \E_{v\sim F_p}[X(v)=j|Z(v,1)=t]\\ \label{eqn:sum_over_probs_case_2}
&= 1 + \E_{v\sim F_p}[X_i(v)|X_i(v)\geq 1],
\end{align}
for any agent $i\in N$. The last equality is true since, under the condition $Z(v,1)=t$, $X(v)$ is the sum $\sum_{i\in S}{X_i(v)}$ for a set $S$ of $t$ agents who are selected uniformly at random and each satisfy $X_i(v)\geq 1$. As the random variables $X_i(v)$ are identically distributed for all agents in $S$, we obtain that $\E[X(v)|Z(v,1)=t]=t\cdot \E[X_i(v)|X_i(v)\geq 1]$ for every agent $i$, which yields the equality.

Now, observe that \begin{align}\label{eqn:exp_num_ones_if_at_least_one_one}
\E_{v\sim F_p}[X_i(v) | X_i(v) \geq 1]
= \frac{\E_{v\sim F_p}[X_i(v)]}{\Pr_{v\sim F_p}[X_i(v)\geq 1]}
= \frac{pm}{1-(1-p)^m}.
\end{align}
The derivative of the RHS of (\ref{eqn:exp_num_ones_if_at_least_one_one}) with respect to $p\in (0,1)$ is
\begin{align*}
  \frac{m}{(1-(1-p)^m)^2}(1-(1-p)^{m-1}(1+p(m-1)) &\geq  \frac{m}{(1-(1-p)^m)^2}\left(1-((1-p)(1+p))^{m-1}\right)\\
&= \frac{m}{(1-(1-p)^m)^2}\left(1-\left(1-p^2\right)^{m-1}\right)>0.
\end{align*}
The first inequality follows from the inequality $(1+t)^r\geq 1+rt$ for $r\geq 1$. Hence, $\E[X_i(v) | X_i(v) \geq 1]$ is strictly increasing in $p$. Since $p<1/m$ in the current case, it follows from (\ref{eqn:exp_num_ones_if_at_least_one_one}) that
\begin{align*}
\E_{v\sim F_p}[X_i(v) | X_i(v) \geq 1]
= \frac{pm}{1-(1-p)^m}
\leq \frac{1}{1-(1-1/m)^m}
\leq \frac{1}{1-e^{-1}},
\end{align*}
using the inequality $(1-1/m)^m \leq e^{-1}$ for any $m\geq 1$.

Using this last observation together with inequalities (\ref{eqn:max_sw_case_2_cond_Z_t}) and (\ref{eqn:sum_over_probs_case_2}), we obtain that 
\begin{align*}
\E_{v\sim F_p}\left[\max_{a\in A} \sw(a,v)\right]
\leq \frac{2e-1}{e-1} \cdot \E_{\substack{v\sim F_p\\P\sim \PP(v)}}\left[\max_{a\in A} \impliedsw(a,P,v,1)\right].
\end{align*}
In combination with the lower bound on the social welfare of the mechanism $\med$ in inequality~(\ref{eqn:sw_mech_case_2}), this yields an average distortion of at most $2.6$.

\paragraph{Case 3.} Again, mechanism $\med$ queries the top-ranked alternative in each agent. Notice that if we have at most two agents, the average distortion of mechanism $\med$ is at most $2$. Indeed, mechanism $\med$ always returns an alternative of positive social welfare whenever there exists one, and no alternative ever has a social welfare higher than $2$. So, in the following, we assume that $n\geq 3$. We will show that the average distortion is lower than $2$ in this case.

Consider impartial culture electorates with $n$ agents and $m$ alternatives with underlying values drawn from the distribution $F_p\in \F_3$. Notice that the maximum social welfare is never larger than the total social welfare of all alternatives. Hence,
\begin{align*}
\E_{v\sim F_p}\left[\max_{a\in A}{\sw(a,v)}\right]\leq pnm.
\end{align*}
With probability $1-(1-p)^{nm}$, there is at least one agent that gives a non-zero value to some alternative. Then, $\med$ returns an alternative of social welfare at least $1$. Hence,
\begin{align*}
\E_{\substack{v\sim F_p\\P\sim \PP(v)}}\left[\sw(\med(P,v),v)\right]\geq 1-(1-p)^{nm}.
\end{align*}
The average distortion of $\med$ can therefore be upper-bounded by the term $\frac{pnm}{1-(1-p)^{nm}}$. Notice that the derivative with respect to $p$ of this quantity is
\begin{align*}
& \frac{nm}{(1-(1-p)^{nm})^2}\cdot (1-(1-p)^{nm-1}(1-p+pnm))\\ 
&\geq\, \frac{nm}{(1-(1-p)^{nm})^2}\cdot (1-(1-p)^{nm-1}(1+p)^{nm-1})\\
&=\, \frac{nm}{(1-(1-p)^{nm})^2}\cdot (1-(1-p^2)^{nm-1})>0.
\end{align*}
The first inequality follows from the property $(1+t)^r\geq 1+rt$ for $r\geq 1$ and the second (strict) inequality is due to the fact that $p > 0$. Hence, the average distortion of the mechanism is strictly increasing in $p$. Using $p^*=1-(1-1/n)^{1/m}$ and since $p<p^*$, we have
\begin{align}\label{eq:ratio-at-p-star}
\adist(\med,\F_3)=\max_{F_p\in \F_3}\frac{\E_{v\sim F_p}\left[\max_{a\in A}{\sw(a,v)}\right]}{\E_{\substack{v\sim F\\P\sim \PP(v)}}\left[\sw(M(P,v),v)\right]}\leq \frac{pnm}{1-(1-p)^{nm}}<\frac{p^*nm}{1-(1-p^*)^{nm}}.
\end{align}
Substituting $p^*$, we get that the denominator in the RHS of (\ref{eq:ratio-at-p-star}) is equal to $1-(1-1/n)^n\geq 1-e^{-1}$. To bound the numerator, observe that $\left(1-\frac{3\ln{\frac{3}{2}}}{nm}\right)^{nm}<\frac{8}{27}$ and that $(1-1/n)^n\geq \frac{8}{27}$ for $n\geq 3$. These inequalities follow since the expression $(1-r/t)^t$ is strictly increasing for $t\geq r$ and approaches $e^{-r}$ from below as $t$ goes to infinity. Thus, we have that $\left(1-\frac{3\ln{\frac{3}{2}}}{nm}\right)^m<1-1/n$, which is equivalent to $\frac{3\ln{\frac{3}{2}}}{nm}>1-(1-1/n)^{1/m}=p^*$, implying that the numerator of the RHS of (\ref{eq:ratio-at-p-star}) is at most $3\ln{\frac{3}{2}}$. We conclude that $\adist(\med,\F_3)<\frac{3\ln{\frac{3}{2}}}{1-e^{-1}}<2$, completing the proof.
\end{proof}

\section{Randomized mechanisms}\label{sec:general-upper-bounds}
We now present two randomized mechanisms for impartial culture electorates with underlying valuations drawn from a {\em general} probability distribution and worst-case electorates, respectively. Both mechanisms are nevertheless similar in spirit. They randomly pick a single threshold from a suitably defined set of thresholds and query each agent to determine a set of alternatives that have value above the threshold. This information is then used to compute an approximation of the alternatives' respective social welfare, which is used to decide the winning alternative.

\subsection{Impartial culture electorates}

Our first ``random threshold'' mechanism \pick\ uses mechanism \med\ as a building block.

\begin{definition}[mechanism \pick]
The mechanism \pick\ uses $k$ thresholds $\ell_1,\ell_2,..., \ell_k$ with $0<\ell_1<...<\ell_k$ as parameters. Given a profile $P$ with underlying valuations drawn from a probability distribution $F$, \pick\ selects an integer $t$ uniformly at random from $[k]$, and sets $p=\Pr_{z\sim F}[z\geq \ell_t]$. It then simulates an execution of \med\ on the distribution $F_p\in\binarydist$ by
\begin{itemize}
\item making the same value queries as \med\ for $F_p$, but
\item interpreting the answer $\val_i(a)$ to a query as $1$ if $\val_i(a)\geq \ell_t$ and $0$ otherwise.
\end{itemize}
\pick\ returns as output the alternative that \med\ selects.
%under the assumption that the given values follow the distribution $F_p$.
\end{definition}

Notice that mechanism \pick\ uses exactly one query per agent.

\begin{theorem}\label{thm:pick}
Let $F$ be a probability distribution over non-negative, real-valued outcomes with mean $\mu$ and variance $\sigma^2$. There is a set of thresholds $\ell_1, \ell_2, ..., \ell_k$ such that the average distortion of mechanism \pick\ is at most $O(\log m + \log \frac{\sigma^2}{\mu^2})$ when applied to impartial culture electorates with $n$ agents and $m$ alternatives, and underlying values drawn according to $F$.
\end{theorem}

\begin{proof}
To prove the theorem, we use the following lemma which relates the average distortion of \pick\ to the structure of the distribution $F$.

\begin{lemma}\label{lem:pick_L_U}
For a random variable $z$ following $F$, assume that there are $L,U>0$ such that
\begin{align*}
    \E_{z\sim F}[z\one\{z<L\}+(z-U)\one\{z\geq U\}] \leq \frac{\mu}{2m}.
\end{align*}
Then, there exists a choice of thresholds $\ell_1, \ell_2, ..., \ell_k$ such that mechanism \pick\ yields average distortion at most $108\left\lceil\log\frac{U}{L}\right\rceil$.
\end{lemma}

\begin{proof}
Set $k=\left\lceil\log\frac{U}{L}\right\rceil$ and define the thresholds of mechanism \pick\ as $\ell_t=L\cdot 2^{t-1}$ for $t=1, 2, ..., k$ and $\ell_k=U$. We begin by observing that, for any $z\geq 0$, we have 
\begin{align}\nonumber
    z & \leq z\one\{z<\ell_1\}+\ell_1\one\{z\geq \ell_1\}+\sum_{t=1}^{k-1}{(\ell_{t+1}-\ell_t)\one\{z\geq \ell_t\}}+(z-\ell_k)\one\{z\geq \ell_k\}\\\label{eq:inequality-on-z}
    & \leq z\one\{z<\ell_1\}+(z-\ell_k)\one\{z\geq \ell_k\}+2\sum_{t=1}^k{\ell_t\one\{z\geq \ell_t\}}.
\end{align}
The second inequality follows since the definition of the thresholds implies that $\ell_{t+1}\leq 2\ell_t$ for $t=1, ..., k-1$ and, hence $\ell_{t+1}-\ell_t\leq \ell_t$. We now have
\begin{align*}
\E_{v\sim F}\left[\max_{a\in A}\sw(a,v)\right]
&= \E_{v\sim F}\left[\max_{a\in A}{\sum_{i=1}^n{\val_i(a)}}\right]\\ 
&\leq \E_{v\sim F}\left[\max_{a\in A}\sum_{i=1}^n{\left(\vphantom{\sum_{t=1}^k}\val_i(a)\one\{\val_i(a)<\ell_1\}+(\val_i(a)-\ell_k)\one\{\val_i(a)\geq \ell_k\}\right.}\right.\\
&\quad\quad \left.\left.+2 \cdot \sum_{t=1}^k{\ell_t\one\{\val_i(a)\geq \ell_t\}}\right)\right]\\
&\leq \E_{v\sim F}\left[\max_{a\in A}\sum_{i=1}^n{\left(\val_i(a)\one\{\val_i(a)<\ell_1\}+(\val_i(a)-\ell_k)\one\{\val_i(a)\geq \ell_k\}\right)}\right]\\
&\quad\quad + 2 \cdot \E_{v\sim F}\left[\max_{a\in A}\sum_{i=1}^n{\sum_{t=1}^k{\ell_t\one\{\val_i(a)\geq \ell_t\}}}\right]\\
&\leq \E_{v\sim F}\left[\sum_{a\in A}\sum_{i=1}^n{\left(\val_i(a)\one\{\val_i(a)<\ell_1\}+(\val_i(a)-\ell_k)\one\{\val_i(a)\geq \ell_k\}\right)}\right]\\
&\quad\quad + 2\cdot \E_{v\sim F}\left[\max_{a\in A}\sum_{i=1}^n{\sum_{t=1}^k{\ell_t\one\{\val_i(a)\geq \ell_t\}}}\right]\\
&\leq \frac{1}{2m}\cdot \E_{v\sim F}\left[\sum_{a\in A}\sum_{i=1}^n{\val_i(a)}\right]+ 2\cdot \E_{v\sim F}\left[\max_{a\in A}\sum_{i=1}^n{\sum_{t=1}^k{\ell_t\one\{\val_i(a)\geq \ell_t\}}}\right]\\
&\leq \frac{1}{2}\cdot \E_{v\sim F}\left[\max_{a\in A}{\sw(a,v)}\right]+ 2\cdot \E_{v\sim F}\left[\max_{a\in A}\sum_{i=1}^n{\sum_{t=1}^k{\ell_t\one\{\val_i(a)\geq \ell_t\}}}\right].
\end{align*}
The first inequality follows from (\ref{eq:inequality-on-z}), the second and third inequalities use the fact that the maximum among non-negative values is upper-bounded by their sum, the fourth inequality uses the assumption in the statement of the lemma for the random variable $\val_i(a)$ with $\mu = \E_{\val_i(a)\sim F}[\val_i(a)]$, and the last inequality follows since the average among non-negative values is a lower-bound on the maximum value among them. The above inequality is equivalent to 
\begin{align}\label{eq:intermediate-guarantee}
 \E_{v\sim F}\left[\max_{a\in A}\sum_{i=1}^n{\sum_{t=1}^k{\ell_t\one\{\val_i(a)\geq \ell_t\}}}\right]
 &\geq \frac{1}{4}\cdot \E_{v\sim F}\left[\max_{a\in A}{\sw(a,v)}\right].
\end{align}

For valuations $v$ and $t\in [k]$, define the valuations $v^t$ that consists of the binary values $\one\{\val_i(a)\geq \ell_t\}$. Notice that for every alternative $a\in A$, it is $\sw(a,v) \geq \ell_t\cdot \sw(a,v^t)$. Mechanism \pick\ selects the integer $t$ uniformly at random from $[k]$ and, for every profile $P$ that is consistent with the valuations $v$, it returns the alternative $\med(P,v^t)$. We thus have
\begin{align*}
    &\E_{\substack{v\sim F\\P\sim\PP(v)}}\left[\sw(\pick(P,v),v)\right]\\
    &= \frac{1}{k}\cdot \sum_{t=1}^k{\E_{\substack{v\sim F\\P\sim\PP(v)}}\left[\sw(\med(P,v^t),v)\right]}\geq \frac{1}{k}\cdot \sum_{t=1}^k{\ell_t\cdot \E_{\substack{v\sim F\\P\sim\PP(v)}}\left[ \sw(\med(P,v^t),v^t)\right]}\\
    &\geq  \frac{1}{27k}\sum_{t=1}^k{\ell_t\cdot \E_{v\sim F}\left[\max_{a\in A}{\sw(a,v^t)}\right]} \geq \frac{1}{27k}\E_{v\sim F}\left[\max_{a\in A}{\sum_{t=1}^k{\ell_t\cdot \sw(a,v^t)}}\right]\\
    &= \frac{1}{27k}\cdot \E_{v\sim F}\left[\max_{a\in A}{\sum_{i=1}^n{\sum_{t=1}^k{\ell_t \one\{\val_i(a)\geq \ell_t\}}}}\right] \geq \frac{1}{108k}\cdot \E_{v\sim F}\left[\max_{a\in A}{\sw(a,v)}\right],
\end{align*}
as desired. The second inequality follows from the average distortion guarantee for mechanism \med\ from Theorem~\ref{thm:mean_mechanism_distortion}, and the fifth one by inequality (\ref{eq:intermediate-guarantee}).
\end{proof}

Now, let $L = \frac{\mu}{4m}$ and $U = \mu + \frac{4m\sigma^2}{\mu}$. Notice that
\begin{equation}\label{eqn:bound_exp_below_L}
\E_{z\sim F}[z\one\{z < L\}]\leq L\cdot \Pr_{z\sim F}[z < L] \leq \frac{\mu}{4m}.
\end{equation}
We slightly overload our notation and denote by $F$ also the cumulative distribution function of a random variable $z$ following the distribution $F$. For $t\geq U$, we have that
$$1-F(t)
= \Pr_{z\sim F}[z\geq t]
\leq \Pr_{z\sim F}\left[|z - \mu| \geq t - \mu \right]
\leq \frac{\sigma^2}{(t-\mu)^2}$$
where the last transition follows from Chebyshev's inequality. Then,
\begin{equation}\label{eqn:bound_exp_above_U}
\E_{z\sim F}\left[(z-U)\one\{z\geq U\}\right]
= \int_U^\infty (1-F(t))\,dt
\leq \int_U^\infty \frac{\sigma^2}{(t-\mu)^2}\,dt
= \frac{\sigma^2}{U-\mu}
= \frac{\mu}{4m}.
\end{equation}
By (\ref{eqn:bound_exp_below_L}) and (\ref{eqn:bound_exp_above_U}), the condition of Lemma~\ref{lem:pick_L_U} is satisfied with $U/L = 4m + 16m^2 \frac{\sigma^2}{\mu^2}$. The theorem then follows from the bound on the average distortion of \pick\ provided by Lemma~\ref{lem:pick_L_U}.
\end{proof}

As an immediate consequence of Theorem~\ref{thm:pick}, we obtain $O(\log{m})$ bounds on the average distortion of \pick\ for fundamental families of distributions. The relevant properties of these distributions which we use for the following corollary can be found in, e.g., \citet[table~5-2, p.~162]{PU02}.

\begin{corollary}
The average distortion of mechanism \pick\ is at most $O(\log{m})$ when applied to impartial culture electorates with $n$ agents and $m$ alternatives, and underlying values drawn according to
\begin{itemize}
\item the exponential distribution $E(\lambda)$ where $\mu = \frac{1}{\lambda}, \sigma^2 = \frac{1}{\lambda^2}$,
\item the chi-squared distribution $\chi^2(k)$ where $\mu = k, \sigma^2 = 2k$, or
\item the Erlang-$k$ distribution $E_k(\lambda)$ where $\mu = \frac{1}{\lambda},\sigma^2=\frac{1}{k\lambda^2}$.
\end{itemize}
In the latter two cases, $k$ denotes a positive integer.
\end{corollary}

We remark that distributions for which our analysis does not give logarithmic average distortion are the $\gamma$-distribution $G(\alpha,\beta)$ where $\mu=\alpha\beta$, $\sigma^2=\alpha\beta^2$ and the $\beta$-distribution $\beta(\alpha,\beta)$ where $\mu = \frac{\alpha}{\alpha+\beta}$, $\sigma^2=\frac{\alpha\beta}{(\alpha+\beta)^2(\alpha+\beta+1)}$ for respective $\alpha$-parameters very close to zero.

\subsection{Worst-case electorates}

Our next random threshold mechanism \rand\ achieves low worst-case distortion.\footnote{We remark that \rand\ is similar in spirit to a mechanism proposed by \citet{BNPS21} in a slightly different participatory budgeting context. There are important differences in modelling assumptions, though. First, their model allows for more powerful queries than ours, where an agent is asked to report all alternatives which she ranks above a certain threshold. More importantly, they assume unit-sum valuations; this assumption affects the design of their mechanism and simplifies its analysis.}

\begin{definition}[mechanism \rand]
For a given profile $P$ with underlying valuations $v$ over $m$ alternatives, mechanism \rand\ picks an integer $r$ uniformly at random from the set $\{1,2,\ldots,\lceil\log 2m\rceil\}$. For every agent $i\in N$, mechanism \rand\
\begin{itemize}
\item queries the value $\nu_i$ of the agent's top-ranked alternative, and
\item finds the set of alternatives $S_{i,r}$, for each of which the agent has value more than $\nu_i/2^r$ using binary search.
\end{itemize}
The mechanism then returns an alternative $\rand(P,v)\in\argmax_{a\in A} \sum_{i\in N} \nu_i \one\{a\in S_{i,r}\}$.
\end{definition}

\begin{theorem}
Mechanism $\rand$ achieves worst-case distortion at most $O(\log m)$ with $O(\log m)$ queries per agent.
\end{theorem}

\begin{proof}
Clearly, for any fixed $r$ and any agent $i$, the alternatives $S_{i,r}$ can be identified by finding the lowest-ranked alternative in $\succ_i$ with value more than $\nu_i/2^r$. This can be accomplished using binary search using $O(\log m)$ queries per agent. For the upper bound on the distortion of \rand, we introduce the following notion of {\em artificial} social welfare. Given valuations $v$ and an alternative $a\in A$, we define the artificial social welfare of $a$ as
$$\artsw(a,v) = \sum_{i\in N} \nu_i \sum_{r=1}^\infty \frac{\one\{a\in S_{i,r}\}}{2^r}.$$
With the next lemma, we show that for any alternative, its artificial social welfare is within a factor of 2 of its social welfare.

\begin{lemma}\label{lemma:artificial_sw}
For any set of valuations $v$ and any alternative $a\in A$, it holds that $\sw(a,v) \leq \artsw(a,v) \leq 2\sw(a,v)$.
\end{lemma}

\begin{proof}
For any agent $i\in N$, let $r_i^*$ be a non-negative integer such that $\val_i(a) \in (\nu_i/2^{r_i^*},\nu_i/2^{r_i^*-1}]$. Notice that thereby
\begin{equation}\label{eqn:artificial_sw_lower}
\sum_{i\in N} \frac{\nu_i}{2^{r_i^*-1}} \geq \sum_{i\in N} \val_i(a) = \sw(a,v),
\end{equation}
and
\begin{equation}\label{eqn:artificial_sw_upper}
\sum_{i\in N}\frac{\nu_i}{2^{r_i^*-1}} < 2\sum_{i\in N} \val_i(a) = 2\sw(a,v).
\end{equation}
Then, for every agent $i$ and any alternative $a$, we have that
$$\sum_{r=1}^\infty \frac{\one\{a\in S_{i,r}\}}{2^r}
= \sum_{r=r_i^*}^\infty \frac{1}{2^r}
= \frac{1}{2^{r_i^*}} \sum_{r=0}^\infty \frac{1}{2^r}
= \frac{1}{2^{r_i^*-1}},$$
and, thus,
\begin{equation*}
\artsw(a,v)
= \sum_{i\in N} \nu_i \sum_{r=1}^\infty \frac{\one\{a\in S_{i,r}\}}{2^r}
= \sum_{i\in N}\frac{\nu_i}{2^{r_i^*-1}}.
\end{equation*}
The lemma follows from inequalities~(\ref{eqn:artificial_sw_lower}) and~(\ref{eqn:artificial_sw_upper}).
\end{proof}
For a given profile $P$, let $a_r$ be the alternative returned by mechanism \rand\ for a draw of $r = 1,2,\ldots,\lceil\log 2m\rceil$, respectively. Denote by $a^*$ the alternative of maximum social welfare for the valuations $v$ underlying $P$. Then, by definition of \rand, it holds that
$$\sum_{i\in N} \frac{\nu_i\one\{a_r\in S_{i,r}\}}{2^r}\geq \sum_{i\in N} \frac{\nu_i\one\{a^*\in S_{i,r}\}}{2^r},$$
for every $r \in \{1,2,\ldots,\lceil\log 2m\rceil\}$, and, hence,
\begin{equation}\label{eqn:sum_over_choice_rand}
\sum_{r=1}^{\lceil\log 2m\rceil}\sum_{i\in N} \frac{\nu_i\one\{a_r\in S_{i,r}\}}{2^r}
\geq \sum_{r=1}^{\lceil\log 2m\rceil}\sum_{i\in N} \frac{\nu_i\one\{a^*\in S_{i,r}\}}{2^r}.
\end{equation}
We now have that
\begin{align}\nonumber 
\sum_{r=1}^{\lceil\log 2m\rceil}\artsw(a_r,v)
&= \sum_{r=1}^{\lceil\log 2m\rceil}\sum_{i\in N}\nu_i\sum_{r'=1}^\infty \frac{\one\{a_r\in S_{i,r'}\}}{2^{r'}} \geq \sum_{i\in N}\nu_i\sum_{r=1}^{\lceil\log 2m\rceil} \frac{\one\{a_r\in S_{i,r}\}}{2^r}\\\nonumber 
& \geq \sum_{i\in N}\nu_i\sum_{r=1}^{\lceil\log 2m\rceil} \frac{\one\{a^*\in S_{i,r}\}}{2^{r}}\nonumber = \artsw(a^*,v) - \sum_{i\in N}\nu_i\sum_{r=\lceil\log 2m\rceil +1}^\infty \frac{\one\{a^*\in S_{i,r}\}}{2^{r}}\\\nonumber
& \geq \artsw(a^*,v) - \sum_{i\in N}\nu_i \left(\frac{1}{2^{\lceil\log{2m}\rceil+1}}\sum_{r=0}^{\infty} \frac{1}{2^r}\right)
\geq \artsw(a^*,v) - \frac{1}{2m}\sum_{i\in N}\nu_i\\\label{eqn:art_sw_opt}
& \geq \frac12 \artsw(a^*,v).
\end{align}
Here, we used~(\ref{eqn:sum_over_choice_rand}) to arrive at the second inequality. The last inequality follows since the average value of the top-ranked alternatives, that is, $(1/m) \sum_{i\in N}\nu_i$, cannot be higher than $\sw(a^*,v)$ and, by Lemma~\ref{lemma:artificial_sw}, $\sw(a^*,v)\leq \artsw(a^*,v)$. Hence, for any set of valuations $v$ and any profile $P\in\PP(v)$, we have that
\begin{align*}
\E_{r\sim\left[\lceil\log 2m\rceil\right]}[\sw(\rand(P,v),v)]
&= \frac{1}{\lceil\log 2m\rceil}\sum_{r=1}^{\lceil\log 2m\rceil} \sw(a_r,v) \geq \frac{1}{2\lceil\log 2m\rceil}\sum_{r=1}^{\lceil\log 2m\rceil} \artsw(a_r,v)\\
&\geq \frac{1}{4\lceil\log 2m\rceil} \artsw(a^*,v) \geq \frac{1}{4\lceil\log 2m\rceil} \sw(a^*,v),
\end{align*}
where the first inequality follows from Lemma~\ref{lemma:artificial_sw}, the second inequality follows from (\ref{eqn:art_sw_opt}), and the third inequality again follows from Lemma~\ref{lemma:artificial_sw}. This concludes the proof of the theorem.
\end{proof}

\section{Worst-case distortion lower bounds}
\label{sec:lower-bounds}

We conclude our technical exposition by presenting two lower bounds on the worst-case distortion. Our basic approach in both of them is as follows. First, for every (large enough) value of $m$ and a value of $n$ of our choice, we decide the agents' rankings. For every position in an agent's ranking, we pre-define a value that is revealed if this particular position is queried by a mechanism. Let $\widehat{a}$ be the alternative that a mechanism picked as winning on the given profile. We then show that ---for every choice of $\widehat{a}$--- it is possible to {\em fix} (i.e., to choose) the agents' remaining concealed valuations in such a way that the distortion is high. That is, for any position not queried by the mechanism, we assume an adversarial set of valuations that is consistent with the agents' rankings and with the values revealed to the mechanism.

\subsection{Lower bounding the number of queries for constant distortion}

Our first lower bound on the number of queries per agent that are necessary to get constant worst-case distortion improves the previously best bound of~\citet{ABFV21} by a sublogarithmic factor. More specifically, \citet{ABFV21} present a lower bound construction and show that any mechanism that uses up to $\lambda$ queries per agent must have worst-case distortion at least $\Omega\left(\frac{1}{\lambda}\cdot m^{\frac{1}{2(\lambda+1)}}\right)$. Thus, in order to get constant worst-case distortion, at least $\Omega\left(\frac{\log{m}}{\log\log{m}}\right)$ queries per agent are necessary. In their follow-up work, \citet{ABFV22a} present an improved construction which yields a higher distortion of $\Omega(m^{1/\lambda})$. Unfortunately, $\lambda$ is now required to be a constant. Therefore, their new construction does not provide any lower bound on the number of queries per agent necessary to get constant worst-case distortion. We prove the next theorem using a considerably different construction which shows that mechanisms making at most $\lambda$ queries per agent have worst-case distortion at least $\frac{1}{8}\cdot m^{\frac{1}{3\lambda}}$ for values of $\lambda$ that are allowed to be logarithmic in $m$.

\begin{theorem}\label{thm:omega_log_m_lower_bound_non_stochastic}
Any deterministic mechanism that achieves a constant worst-case distortion must make $\Omega(\log m)$ queries per agent.
\end{theorem}

\begin{proof}
Let $m\geq 154$ and $\lambda$ be an integer such that $2\leq \lambda\leq \log{m}$. Consider a mechanism $\mechanism$ that makes at most $\lambda$ queries per agent; we will show that $\mechanism$ has worst-case distortion at least $\frac{1}{8}m^{\frac{1}{3\lambda}}$.

We define the symmetric profile $P=\{\succ_i\}_{i\in N}$ with $n=m$ agents, so that agent $i$ has the ranking
$$i \succ_i i+1 \succ_i \dots \succ_i m \succ_i 1 \succ_i \ldots \succ_i i-1.$$
The ranking of every agent $i$ is divided into $2\lambda + 1$ sets ---or {\em buckets}--- $B^{(i)}_1, \ldots, B^{(i)}_{2\lambda + 1}$ where
$$\left|B^{(i)}_j\right| = b_j = \left\lceil m^{\frac{j}{3 \lambda}}\right\rceil,$$
for $j \in [2\lambda]$, and $|B^{(i)}_{2\lambda+1}| = b_{2\lambda + 1} = m - \sum_{j=1}^{2\lambda} b_j$. Hence, $$B_j^{(i)}=\left\{i+\sum_{t=1}^{j-1}{b_t}\bmod m, ..., i-1+\sum_{t=1}^{j}{b_t}\bmod m\right\}.$$ 
We refer to the alternatives from bucket $B^{(i)}_{2\lambda + 1}$ as the {\em tail alternatives} of agent $i$.\footnote{Our assumptions $m\geq 154$ and $\lambda\leq \log{m}$ guarantee that $\sum_{j=1}^{2\lambda}{b_j}\leq m$ and, thus, buckets $B_{2\lambda+1}^{(i)}$ are well-defined.}

We proceed to describe the agents' valuations $v$. Every agent $i$ assigns a value of $0$ to each of her tail alternatives, i.e., $\val_i(a) = 0$ for every $a \in  B^{(i)}_{2\lambda + 1}$. For $j\in [2\lambda]$, agent $i$ assigns to all the alternatives of bucket $B_j^{(i)}$ either a {\em low} value of $m^{\frac{2\lambda - j}{3 \lambda}}$ or a {\em high} value of $m^{\frac{2\lambda - j + 1}{3 \lambda}}$ in the following way. Whenever the mechanism $\mechanism$ makes a query for the value of an alternative in bucket $B_j^{(i)}$, the concealed value of each alternative in bucket $B_j^{(i)}$ is set to the low value, i.e.,  $\val_i(a)=m^{\frac{2\lambda-j}{3\lambda}}$ for every alternative $a\in B_j^{(i)}$; this value is also revealed as the outcome of the query. Now, consider a bucket $B_j^{(i)}$, in which mechanism $\mechanism$ did not query the value of any alternative. The concealed values of all alternatives in this bucket are set to the low value $m^{\frac{2\lambda - j}{3 \lambda}}$ if the winning alternative $\mechanism(P,v)$ belongs to the bucket and the high value $m^{\frac{2\lambda - j+1}{3 \lambda}}$ otherwise. Figure~\ref{fig:refined_lower_bound_worst_case} shows an example that demonstrates this approach.

\begin{figure}[t]
\centering
\if\drawfigures1
\begin{tikzpicture}
\begin{scope}[scale=0.15,every node/.style={align=center,scale=0.85}]
    % queries and high value buckets
    \fill[lightgray](0,-36) rectangle (3,-47);
    \fill[lightgray](4,0) rectangle (10,-3);
    \fill[lightgray](4,-33) rectangle (10,-47);
    \fill[lightgray](11,0) rectangle (23,-4);
    \fill[lightgray](11,-12) rectangle (23,-15);
    \fill[lightgray](11,-27) rectangle (23,-32);
    \fill[lightgray](11,-44) rectangle (23,-47);
    \fill[lightgray](24,-15) rectangle (46,-27);
    \fill[lightgray](24,-30) rectangle (46,-40);
    % queries
    \draw[pattern={Lines[angle=-45,distance=2pt]}](0,0) rectangle (1,-33);
    \draw[pattern={Lines[angle=-45,distance=2pt]}](4,-3) rectangle (5,-30);
    \draw[pattern={Lines[angle=-45,distance=2pt]}](11,-4) rectangle (12,-12);
    \draw[pattern={Lines[angle=-45,distance=2pt]}](11,-17) rectangle (12,-23);
    \draw[pattern={Lines[angle=-45,distance=2pt]}](11,-32) rectangle (12,-44);
    \draw[pattern={Lines[angle=-45,distance=2pt]}](24,-27) rectangle (25,-30);
    \draw[pattern={Lines[angle=-45,distance=2pt]}](24,-42) rectangle (25,-47);
    % cohorts
    \draw(0,0) rectangle (3,-47);
    \draw(4,0) rectangle (10,-47);
    \draw(11,0) rectangle (23,-47);
    \draw(24,0) rectangle (46,-47);
    \draw(47,0) rectangle (51,-47);
    % a hat = 36 alternative picked by the mechanism
    \foreach \i in {1,...,3}{
      \foreach \j in {1,...,3}{
        \fill[black](\i - 1,-36 +\i -1) rectangle (\i,-36 + \i);
      }
    }
    \foreach \i in {1,...,6}{
      \foreach \j in {1,...,6}{
        \fill[black](4 + \i - 1,-33 +\i -1) rectangle (4 + \i,-33 + \i);
      }
    }
    \foreach \i in {1,...,12}{
      \foreach \j in {1,...,12}{
        \fill[black](11 + \i - 1,-27 +\i -1) rectangle (11 + \i,-27 + \i);
      }
    }
    \foreach \i in {1,...,15}{
      \foreach \j in {1,...,15}{
        \fill[black](24 + \i - 1,-15 +\i -1) rectangle (24 + \i,-15 + \i);
      }
    }
    \foreach \i in {1,...,7}{
      \foreach \j in {1,...,7}{
        \fill[black](39 + \i - 1,-47 +\i -1) rectangle (39 + \i,-47 + \i);
      }
    }
    \foreach \i in {1,...,4}{
      \foreach \j in {1,...,4}{
        \fill[black](47 + \i - 1,-40 +\i -1) rectangle (47 + \i,-40 + \i);
      }
    }
\end{scope}
\end{tikzpicture}
\fi
\caption{An example for our lower bound construction in Theorem~\ref{thm:omega_log_m_lower_bound_non_stochastic} that illustrates the way in which the agents' valuations are defined. The alternative $\widehat{a}$ picked by the mechanism $\mechanism$ is marked as a black box in every agent's ranking. We assume that $\mechanism$ queried the positions corresponding to the dashed boxes. In this example, the mechanism only queries the first position in a bucket which is without loss of generality. The gray areas correspond to buckets in which all alternatives have high values. White areas correspond to buckets in which alternatives have low values, either because the bucket contains the winning alternative $\widehat{a}$ or because mechanism $\mechanism$ queried the value of an alternative in the bucket.}
\label{fig:refined_lower_bound_worst_case}
\end{figure}
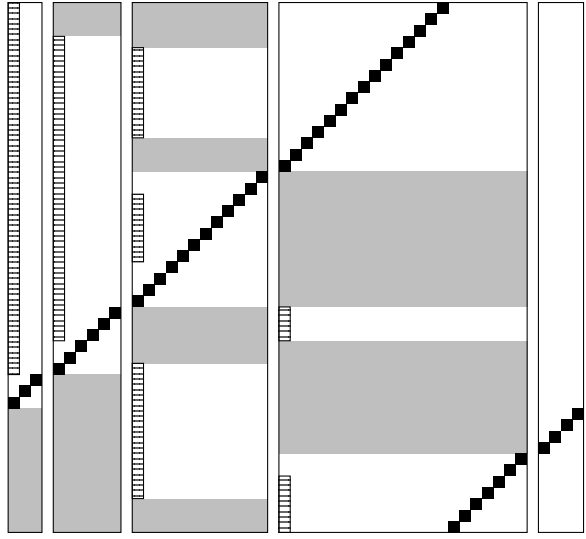

Let $\widehat{a}=\mechanism(P,v)$. Observe that alternative $\widehat{a}$ belongs to bucket $B_j^{(i)}$ for $b_j$ different choices of $i\in N$. Hence,
\begin{align}\label{eq:sw-winner}
\sw(\widehat{a},v)
&= \sum_{i \in N}\,\sum_{j \in [2\lambda]: \widehat{a} \in B_j^{(i)}} m^{\frac{2\lambda - j}{3 \lambda}}
= \sum_{j \in [2\lambda]} b_j \cdot m^{\frac{2\lambda - j}{3 \lambda}}
= \sum_{j \in [2\lambda]} \left\lceil m^{\frac{j}{3 \lambda}}\right\rceil \cdot m^{\frac{2\lambda - j}{3 \lambda}}
\leq 4\lambda m^{2/3}.
\end{align}

We now compute the sum of the social welfare over all alternatives by summing up all the values in every bucket of every agent. To do so, define the subsets $H$ and $L$ of $N\times [2\lambda]$ as follows. $L$ consists of the pairs $(i,j)$ such that either $\mechanism(P,v) \in B_j^{(i)}$ or mechanism $\mechanism$ queried the value of some alternative in bucket $B_j^{(i)}$. Let $H=N\times [2\lambda]\setminus L$. We have
\begin{align}\nonumber
\sum_{a\in A}{\sw(a,v)} &= \sum_{i\in N}{\left(\sum_{j\in [2\lambda]:(i,j)\in L}{b_j\cdot m^{\frac{2\lambda - j}{3 \lambda}}}+\sum_{j\in [2\lambda]:(i,j)\in H}{b_j\cdot m^{\frac{2\lambda - j+1}{3 \lambda}}}\right)}\\\nonumber
&= \sum_{i\in N}{\left(\sum_{j\in [2\lambda]}{b_j\cdot m^{\frac{2\lambda - j}{3 \lambda}}}+\sum_{j\in [2\lambda]:(i,j)\in H}{b_j\cdot \left(m^{\frac{2\lambda - j+1}{3 \lambda}}-m^{\frac{2\lambda - j}{3 \lambda}}\right)}\right)}\\\label{eq:sum-of-sws}
&\geq \sum_{i\in N}{\sum_{j\in [2\lambda]:(i,j)\in H}{b_j\cdot m^{\frac{2\lambda - j+1}{3 \lambda}}}} \geq \sum_{i\in N}{\sum_{j\in [2\lambda]:(i,j)\in H}{m^{\frac{2\lambda +1}{3 \lambda}}}} =|H|\cdot m^{\frac{2\lambda +1}{3 \lambda}}.
\end{align}
Now, observe that for every agent $i$, the set $L$ contains at most $\lambda+1$ pairs $(i,\cdot)$ for the up to $\lambda$ buckets in which $\mechanism$ queried the value of some alternative and (possibly) one extra bucket that contains alternative $\widehat{a}$. Hence, $|L|\leq (\lambda+1)\cdot n$ and, consequently, $|H|\geq (\lambda-1)\cdot n$. Recalling that $n=m$, inequality (\ref{eq:sum-of-sws}) yields
\begin{align}\label{eq:max-sw}
\max_{a\in A}{\sw(a,v)} &\geq \frac{1}{m}\cdot \sum_{a\in A}{\sw(a,v)} \geq (\lambda-1)\cdot m^{\frac{2\lambda +1}{3 \lambda}}.
\end{align}
The desired lower bound of $\frac{\lambda-1}{4\lambda}\cdot m^{\frac{1}{3\lambda}}\geq \frac{1}{8}\cdot m^{\frac{1}{3\lambda}}$ on the distortion now follows from inequalities (\ref{eq:sw-winner}) and (\ref{eq:max-sw}).
\end{proof}

\subsection{Lower bounding the distortion of 1-query mechanisms}\label{sec:worst_case_lower_bound_1_query_binary}

In Section~\ref{sec:1-query}, we saw that a single query per agent is sufficient to guarantee constant average distortion when the agents draw their valuations according to a binary distribution. Such a guarantee is not attainable in the traditional setting of worst-case distortion. Indeed, \citet{ABFV21} proved that any deterministic 1-query mechanism must have distortion $\Omega(m)$. However, their lower bound construction uses valuations that are more complex than binary. Our next result states that, even in the case of binary valuations, the worst-case distortion of deterministic $1$-query mechanisms must still be high.

\begin{theorem}\label{thm:0/1_valued_1_query_lower_bound}
Every deterministic $1$-query mechanism has a worst-case distortion of at least $\Omega(\sqrt{m})$. This is true even if the agents have binary values for each of the alternatives.
\end{theorem}

\begin{proof}
Let $m\geq 16$ and $t$ be the largest even integer such that $t^2\leq m$. Clearly, $t\in\Omega(\sqrt{m})$. We consider the profile $P = \{\succ_i\}_{i\in N}$ with $n=t^2$ agents so that for every agent $i \in N$, the ranking $\succ_i$ has the form
$$i \succ_i i+1 \succ_i \dots \succ_i t^2 \succ_i 1 \succ_i ... \succ_i i-1 \succ_i t^2+1 \succ_i \ldots \succ_i m.$$
We divide the $t^2$ agents into $t$ groups, each containing $t$ agents. We call these groups {\em cohorts}. For $k \in [t]$, the $k$-th cohort $\cohort_k$ contains the agents $(k-1)t +1,\ldots, kt$. Due to the symmetry of $P$ and the assumption that $n=t^2$, an element in $\cohort_k$ may refer to an agent $i\in\cohort_k$ as well as to an alternative $j$ that is the top-ranked alternative of agent $j\in\cohort_k$.  Figure~\ref{fig:01_1_query_lower_bound_construction} shows an example of our lower bound construction.

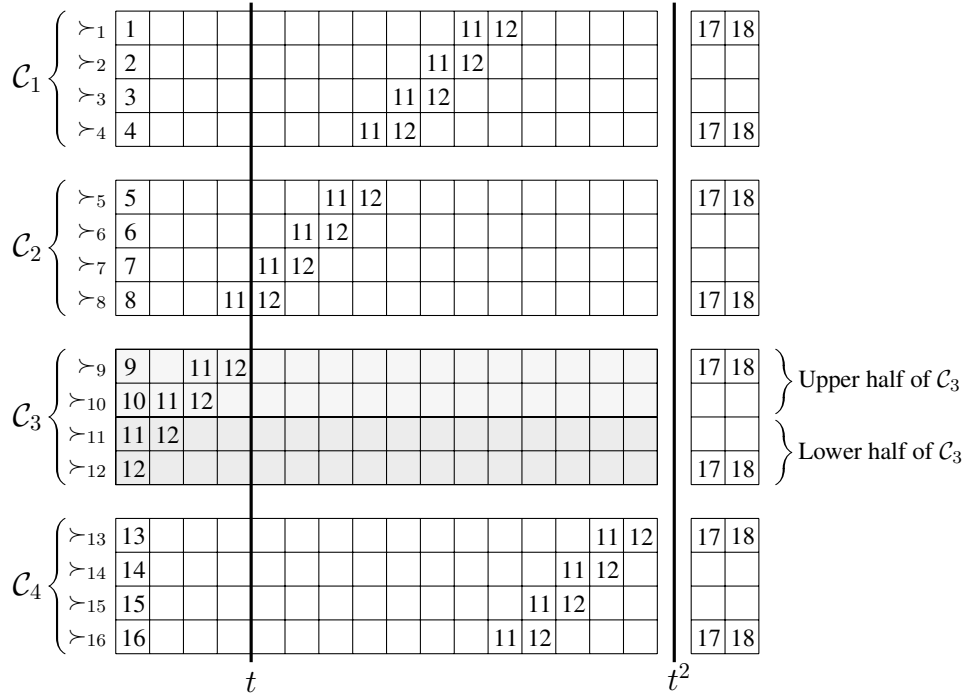
\begin{figure}[t]
\centering
\if\drawfigures1
\begin{tikzpicture}
\begin{scope}[scale=0.45,every node/.style={align=center,scale=0.85}]
    % cohorts
    \draw (0, 0) grid (16,-4);
    \draw (0,-5) grid (16,-9);
    \draw[fill=lightgray!15] (0,-10) grid (16,-12) rectangle (0,-10);
    \draw[fill=lightgray!30] (0,-12) grid (16,-14) rectangle (0,-12);
    \draw (0,-15) grid (16,-19);
    % Cohort 1
    \draw [decorate,decoration={calligraphic brace,amplitude=6pt},thick] (-1.5,-4) -- (-1.5,0) node[left=6pt,pos=.5]{\Large $\cohort_1$};
    \draw (0,-1) node[left,yshift=.25cm]{$\succ_1$};
    \draw (0,-2) node[left,yshift=.25cm]{$\succ_2$};
    \draw (0,-3) node[left,yshift=.25cm]{$\succ_3$};
    \draw (0,-4) node[left,yshift=.25cm]{$\succ_4$};
    \draw (0,-1) node[right,yshift=.25cm]{1};
    \draw (10,-1) node[right,yshift=.25cm,xshift=-.05cm]{11};
    \draw (11,-1) node[right,yshift=.25cm,xshift=-.05cm]{12};
    \draw (0,-2) node[right,yshift=.25cm]{2};
    \draw (9,-2) node[right,yshift=.25cm,xshift=-.05cm]{11};
    \draw (10,-2) node[right,yshift=.25cm,xshift=-.05cm]{12};
    \draw (0,-3) node[right,yshift=.25cm]{3};
    \draw (8,-3) node[right,yshift=.25cm,xshift=-.05cm]{11};
    \draw (9,-3) node[right,yshift=.25cm,xshift=-.05cm]{12};
    \draw (0,-4) node[right,yshift=.25cm]{4};
    \draw (7,-4) node[right,yshift=.25cm,xshift=-.05cm]{11};
    \draw (8,-4) node[right,yshift=.25cm,xshift=-.05cm]{12};
     % Cohort 2
     \draw [decorate,decoration={calligraphic brace,amplitude=6pt},thick] (-1.5,-9) -- (-1.5,-5) node[left=6pt,pos=.5]{\Large $\cohort_2$};
    \draw (0,-6) node[left,yshift=.25cm]{$\succ_5$};
    \draw (0,-7) node[left,yshift=.25cm]{$\succ_6$};
    \draw (0,-8) node[left,yshift=.25cm]{$\succ_7$};
    \draw (0,-9) node[left,yshift=.25cm]{$\succ_8$};
    \draw (0,-6) node[right,yshift=.25cm]{5};
    \draw (6,-6) node[right,yshift=.25cm,xshift=-.05cm]{11};
    \draw (7,-6) node[right,yshift=.25cm,xshift=-.05cm]{12};
    \draw (0,-7) node[right,yshift=.25cm]{6};
    \draw (5,-7) node[right,yshift=.25cm,xshift=-.05cm]{11};
    \draw (6,-7) node[right,yshift=.25cm,xshift=-.05cm]{12};
    \draw (0,-8) node[right,yshift=.25cm]{7};
    \draw (4,-8) node[right,yshift=.25cm,xshift=-.05cm]{11};
    \draw (5,-8) node[right,yshift=.25cm,xshift=-.05cm]{12};
    \draw (0,-9) node[right,yshift=.25cm]{8};
    \draw (3,-9) node[right,yshift=.25cm,xshift=-.05cm]{11};
    \draw (4,-9) node[right,yshift=.25cm,xshift=-.05cm]{12};
    % Cohort 3
    \draw [decorate,decoration={calligraphic brace,amplitude=6pt},thick] (-1.5,-14) -- (-1.5,-10) node[left=6pt,pos=.5]{\Large $\cohort_3$};
    \draw [decorate,decoration={calligraphic brace,amplitude=6pt},thick] (19.5,-10) -- (19.5,-11.9) node[right=6pt,pos=.5]{Upper half of $\cohort_3$};
    \draw [decorate,decoration={calligraphic brace,amplitude=6pt},thick] (19.5,-12.1) -- (19.5,-14) node[right=6pt,pos=.5]{Lower half of $\cohort_3$};
    \draw (0,-11) node[left,yshift=.25cm]{$\succ_9$};
    \draw (0,-12) node[left,yshift=.25cm]{$\succ_{10}$};
    \draw (0,-13) node[left,yshift=.25cm]{$\succ_{11}$};
    \draw (0,-14) node[left,yshift=.25cm]{$\succ_{12}$};
    \draw (0,-11) node[right,yshift=.25cm]{9};
    \draw (2,-11) node[right,yshift=.25cm,xshift=-.05cm]{11};
    \draw (3,-11) node[right,yshift=.25cm,xshift=-.05cm]{12};
    \draw (0,-12) node[right,yshift=.25cm,xshift=-.05cm]{10};
    \draw (1,-12) node[right,yshift=.25cm,xshift=-.05cm]{11};
    \draw (2,-12) node[right,yshift=.25cm,xshift=-.05cm]{12};
    \draw (0,-13) node[right,yshift=.25cm,xshift=-.05cm]{11};
    \draw (1,-13) node[right,yshift=.25cm,xshift=-.05cm]{12};
    \draw (0,-14) node[right,yshift=.25cm,xshift=-.05cm]{12};
    % Cohort 4
    \draw [decorate,decoration={calligraphic brace,amplitude=6pt},thick] (-1.5,-19) -- (-1.5,-15) node[left=6pt,pos=.5]{\Large $\cohort_4$};
    \draw (0,-16) node[left,yshift=.25cm]{$\succ_{13}$};
    \draw (0,-17) node[left,yshift=.25cm]{$\succ_{14}$};
    \draw (0,-18) node[left,yshift=.25cm]{$\succ_{15}$};
    \draw (0,-19) node[left,yshift=.25cm]{$\succ_{16}$};
    \draw (0,-16) node[right,yshift=.25cm,xshift=-.05cm]{13};
    \draw (14,-16) node[right,yshift=.25cm,xshift=-.05cm]{11};
    \draw (15,-16) node[right,yshift=.25cm,xshift=-.05cm]{12};
    \draw (0,-17) node[right,yshift=.25cm,xshift=-.05cm]{14};
    \draw (13,-17) node[right,yshift=.25cm,xshift=-.05cm]{11};
    \draw (14,-17) node[right,yshift=.25cm,xshift=-.05cm]{12};
    \draw (0,-18) node[right,yshift=.25cm,xshift=-.05cm]{15};
    \draw (12,-18) node[right,yshift=.25cm,xshift=-.05cm]{11};
    \draw (13,-18) node[right,yshift=.25cm,xshift=-.05cm]{12};
    \draw (0,-19) node[right,yshift=.25cm,xshift=-.05cm]{16};
    \draw (11,-19) node[right,yshift=.25cm,xshift=-.05cm]{11};
    \draw (12,-19) node[right,yshift=.25cm,xshift=-.05cm]{12};
    % Separators
    \draw[very thick] (4,.25) -- (4,-19.25) node[below]{\Large $t$};
    \draw[very thick] (16.5,.25) -- (16.5,-19.25) node[below,yshift=.12cm,xshift=.08cm]{\Large $t^2$};
    % Tail alternatives
    \draw (17, 0) grid (19,-4);
    \draw (17,-1) node[right,yshift=.25cm,xshift=-.05cm]{17};
    \draw (18,-1) node[right,yshift=.25cm,xshift=-.05cm]{18};
    \draw (17,-4) node[right,yshift=.25cm,xshift=-.05cm]{17};
    \draw (18,-4) node[right,yshift=.25cm,xshift=-.05cm]{18};
    \draw (17,-5) grid (19,-9);
    \draw (17,-6) node[right,yshift=.25cm,xshift=-.05cm]{17};
    \draw (18,-6) node[right,yshift=.25cm,xshift=-.05cm]{18};
    \draw (17,-9) node[right,yshift=.25cm,xshift=-.05cm]{17};
    \draw (18,-9) node[right,yshift=.25cm,xshift=-.05cm]{18};
    \draw (17,-10) grid (19,-14);
    \draw (17,-11) node[right,yshift=.25cm,xshift=-.05cm]{17};
    \draw (18,-11) node[right,yshift=.25cm,xshift=-.05cm]{18};
    \draw (17,-14) node[right,yshift=.25cm,xshift=-.05cm]{17};
    \draw (18,-14) node[right,yshift=.25cm,xshift=-.05cm]{18};
    \draw (17,-15) grid (19,-19);
    \draw (17,-16) node[right,yshift=.25cm,xshift=-.05cm]{17};
    \draw (18,-16) node[right,yshift=.25cm,xshift=-.05cm]{18};
    \draw (17,-19) node[right,yshift=.25cm,xshift=-.05cm]{17};
    \draw (18,-19) node[right,yshift=.25cm,xshift=-.05cm]{18};
\end{scope}
\end{tikzpicture}
\fi
\caption{An example of our lower bound construction in Theorem~\ref{thm:0/1_valued_1_query_lower_bound} where $m = 18$. Hence, we set $t=4$ and $n=t^2=16$ resulting in the profile $P$ shown above. In every agent's ranking (horizontal bars), the example mentions the top-ranked alternative. Additionally, we marked alternatives $11$ and $12$ in every ranking in order to showcase the symmetry of $P$. Alternatives $17$ and $18$ are shown for two agents of every cohort. Notice that these alternatives appear in the exact same positions of every agent's ranking.}
\label{fig:01_1_query_lower_bound_construction}
\end{figure}

Let $\val_{i,j}$ denote the value that agent $i$ has for the alternative in the $j$-th position of her ranking. For every query that the mechanism makes, we reveal the following information:
\begin{itemize}
\item If agent $i$ is the first agent of cohort $\cohort_k$ who is queried by the mechanism at a position $j \leq t$, we reveal $\val_{i,j} = 1$.
\item Otherwise, we reveal $\val_{i,j} = 0$.
\end{itemize}
The latter item includes the case where the mechanism queries an agent at any position $j > t$ as well as the case where the mechanism already queried an agent of cohort $\cohort_k$ at a position $j \leq t$. 

With the next two lemmas, we distinguish two cases that ---taken together--- cover all possible choices a mechanism may make for querying positions in $P$. In each case, we are able to fix the agents' valuations in a way that is consistent with the values revealed to the mechanism and results in a high distortion. The following piece of notation will be helpful for this purpose. Let $N_{\leq j}(a)$ be the set of agents that rank alternative $a$ at position $j$ or higher, that is,
$$N_{\leq j}(a)=\left\{i\in N \,:\, \pos_{\succ_i}(a) \leq j\right\}.$$

\begin{lemma}\label{claim:01_1_query_lower_bound_query_every_cohort}
Assume that there is a cohort $\cohort_k$ such that mechanism $\mechanism$ does not query any agent $i \in \cohort_k$ at a position $j \leq t$. Then, there exists a choice of valuations $v$ that is consistent with profile $P$ and the values revealed to $\mechanism$, such that
$$\frac{\max_{a \in A} \sw(a,v)}{\sw(\mechanism(P,v),v)} \geq t/2.$$
\end{lemma}

\begin{proof}
First, assume that $\mechanism(P,v) = a \notin \cohort_k$. Consider the alternative $a' = kt \in \cohort_k$. Note that $N_{\leq t}(a') = \cohort_k$, that is, all agents of cohort $\cohort_k$ rank alternative $a'$ at a position $j \leq t$. Since the algorithm did not query any of these positions, we are free to fix the concealed values for every $i \in N_{\leq t}(a')$ such that
$$\val_{i,j}=
\begin{cases}
1 & \text{for every position $j \leq \pos_{\succ_i}(a')$}\\
0 & \text{otherwise.}
\end{cases}$$
The remaining concealed values (outside of cohort $\cohort_k$) are set to $0$ except for those positions where a value of $1$ is implied by a revealed value of $1$ at a position further down in the ranking. Thereby, $\sw(a',v) = t$. Furthermore, notice that the set $N_{\leq t}(a)$ intersects with at most two cohorts and alternative $a$ receives a value of $0$ from every agent $i\in N_{\leq t}(a')$. Thus, $\sw(a,v)\leq 2$, which shows that the lemma holds for $\mechanism(P,v) = a \notin \cohort_k$.

Now, let $\mechanism(P,v) = a \in \cohort_k$. We will further distinguish between the cases where $a \leq (k-1/2)t$ (i.e., $a$ is in the {\em upper half} of $\cohort_k$; see Figure~\ref{fig:01_1_query_lower_bound_construction}) and $a > (k-1/2)t$ (i.e., $a$ is in the {\em lower half} of $\cohort_k$). For both cases, we first show that there is an alternative $a'\neq a$ such that $\sw(a',v) \geq t/2$.
\medskip

{\bf Case 1: $a \leq (k-1/2)t$.} Consider alternative $a' = kt$. Note that $N_{\leq t}(a') = \cohort_k$. Furthermore, since $a$ is in the upper half of the cohort, it holds that
$$|N_{\leq t}(a') \setminus N_{\leq t}(a)| \geq t/2.$$
Since $\mechanism$ did not query any agent from $\cohort_k$ at a position $j \leq t$, we can fix the concealed values for every $i\in N_{\leq t}(a') \setminus N_{\leq t}(a)$ such that
$$\val_{i,j}=
\begin{cases}
1 & \text{for every position $j \leq \pos_{\succ_i}(a')$}\\
0 & \text{otherwise.}
\end{cases}$$
Thereby, $\sw(a',v) \geq t/2$.
\medskip

{\bf Case 2: $a > (k-1/2)t$.} Let $a' = a-1$ and note that
$$|N_{\leq t}(a') \cap \cohort_k| \geq t/2.$$
Since $\mechanism$ did not query any agent of cohort $\cohort_k$ at a position $j \leq t$, we are free to fix the concealed values for every agent $i \in N_{\leq t}(a') \cap \cohort_k$ such that
$$\val_{i,j}=
\begin{cases}
1 & \text{for every position $j \leq \pos_{\succ_i}(a')$}\\
0 & \text{otherwise.}
\end{cases}
$$
Thereby, $\sw(a',v) \geq t/2$ in this case as well.
\medskip

In both cases, the remaining concealed values are set to $0$ except for those positions where a value of $1$ is implied by a revealed value of $1$ at a position further down in the ranking. In particular, this means that $\val_i(a) = 0$ for every agent $i \in N_{\leq t}(a) \cap \cohort_k$. By assumption, $\mechanism$ did not query any of these agents at a position $j\leq t$. In case 1, among cohort $\cohort_k$, only the agents $N_{\leq t}(a') \setminus N_{\leq t}(a) = \cohort_k \setminus N_{\leq t}(a)$ assign a value of $1$ to any alternative. In case 2, among cohort $\cohort_k$, only the agents $N_{\leq t}(a') \cap \cohort_k$ have a value of $1$ for any alternative and exclusively for those alternatives ranked above $a$. Finally, by our construction, there is at most one other cohort $k'$ such that $N_{\leq t}(a) \cap \cohort_{k'} \neq \emptyset$. Hence, $\sw(a,v) \leq 1$ and the lemma follows.
\end{proof}

\begin{lemma}\label{claim:01_1_query_lower_bound_omega_sqrt_m}
Assume that mechanism $\mechanism$ queries every cohort at at least one position $j \leq t$. Then, there exists a choice of valuations $v$ that is consistent with profile $P$ and the values revealed to $\mechanism$, such that
$$\frac{\max_{a \in A} \sw(a,v)}{\sw(\mechanism(P,v),v)} \geq t/2-1.$$
\end{lemma}

\begin{proof}
Since no agent gives any value to alternatives $t^2+1, ..., m$, the distortion is infinite when $\mechanism(P,v) > t^2$. Hence, assume that there is a cohort $\cohort_k$ such that $\mechanism(P,v) = a \in \cohort_k$. Consider the alternative $a' = a-1 \mod t^2$. The set of agents $N_{\leq t}(a)$ and $N_{\leq t}(a')$ intersect with at most two cohorts, namely, cohort $k$ and cohort $k' = k-1 \mod t^2$. For every cohort $\cohort_{\ell}$ where $\ell \neq k,k'$, it holds by our assumption that there is an agent $i_{\ell}$ who is the {\em first agent in this cohort} to be queried by $\mechanism$ at a position $j_{\ell}\leq t$. At this position, there is an alternative other than $a$ or $a'$, and we revealed the value $\val_{i_{\ell},j_{\ell}}=1$; see above. Since $\mechanism$ can make at most one query to each agent, we can set the remaining concealed values such that
$$\val_{i_{\ell},j}=
\begin{cases}
1 & \text{for every position $j \leq \pos_{\succ_{i_\ell}}(a')$}\\
0 & \text{otherwise}
\end{cases}$$
for every $\ell \neq k,k'$. This implies that $\val_{i_{\ell}}(a) = 0$ for every $\ell \neq k,k'$. The remaining concealed values are set to 0 except for those positions in $\cohort_k, \cohort_{k'}$ where a value of 1 is implied by a revealed value of 1 at a position further down in the ranking. Then, $\sw(a,v)\leq 2$ since $N_{\leq t}(a)$ can intersect only with cohorts $\cohort_k$ and $\cohort_{k'}$. On the other hand, by setting the concealed values for every cohort $\ell \neq k,k'$ as described, it holds that $\sw(a',v) \geq t-2$. From this, the lemma follows.
\end{proof}

Theorem~\ref{thm:0/1_valued_1_query_lower_bound} now follows by combining Lemmas~\ref{claim:01_1_query_lower_bound_query_every_cohort} and~\ref{claim:01_1_query_lower_bound_omega_sqrt_m}.
\end{proof}

\section{Discussion and open problems}
We have initiated the study of average distortion in a simple stochastic setting that creates impartial culture electorates. The main open problem is whether constant average distortion is possible with a small number of queries per agent for general probability distributions of valuations. Throughout the paper, we assume that the distribution is given as part of the input, and this information is crucial to make our mechanisms work. It would be interesting to explore whether this ---admittedly strong--- requirement can be removed. Other natural extensions of our model include different distributions per alternative or distributions that produce random valuations satisfying the unit-sum or unit-range assumption. These latter assumptions are clearly beyond the reach of our current analysis techniques, as they necessarily imply correlations between the random values an agent has for the alternatives.

For the worst-case setting, we have improved the previously best-known lower bound on the number of queries per agent that are necessary for constant worst-case distortion by deterministic mechanisms. Still, the conjecture of~\citet{ABFV21} that constant worst-case distortion is possible with $\Theta(\log{m})$ deterministic queries per agent is wide open. Furthermore, we have also demonstrated that the use of randomization can yield worst-case bounds that the known deterministic mechanisms cannot achieve. Exploring whether there is a separation between deterministic and randomized mechanisms in terms of their worst-case distortion for a given number of queries per agent is another challenging problem that deserves investigation.

% Bibliography
\bibliography{distortion-ice}

% Appendix
\newpage
\appendix

\section{Proof of Lemma~\ref{lem:balanced}}\label{app:sec:lem:balanced}
For the proof of Lemma~\ref{lem:balanced}, we present another technical statement as Lemma~\ref{lem:t-t-primes}. The latter lemma formalizes the intuition that the expectation for the maximum social welfare only increases when a given number of values of 1 is distributed more evenly between two agents. The proof of the lemma leans heavily on notation. We therefore highlight two key properties that the proof of Lemma~\ref{lem:t-t-primes} exploits. These properties depend crucially on the fact that the agents' rankings are uniformly random and independent. Let $r$ be any non-negative integer.
\begin{itemize}
\item {\bf Property 1}: Let $u = (u_1,u_2,...,u_n)$ be a vector with non-negative integer entries. Now, consider the question whether, for a random draw of $v\sim F$, $P \sim \PP(v)$, there exists an alternative $a\in A$ that appears in the rankings of at least $r$ agents such that for each of these agents $i$ it holds that $\pos_{\succ_i}(a) \leq u_i$. Whether or not such an alternative exists does not depend on the number of values of 1 underlying the agents' rankings.
\item {\bf Property 2}: Let $a_1$ be {\em any} alternative appearing in {\em any} position in the ranking of agent 1 and let $a_2$ be {\em any} alternative appearing in {\em any} position in the ranking of agent 2. The probability of having social welfare of exactly $r$ for agents $3,...,n$ is exactly the same for $a_1$ and $a_2$.
\end{itemize}
We continue with the statement and formal proof of Lemma~\ref{lem:t-t-primes}.
\begin{lemma}\label{lem:t-t-primes}
Consider two vectors $t=(t_1,t_2,..., t_n)$ and $t'=(t'_1,t'_2,..., t'_n)$ with non-negative integer entries such that $t_j\geq t_k+2$ for two entries $j,k\in [n]$ and $t'_j = t_j -1, t'_k = t_k +1$ and $t'_i = t_i$ for all $i \in [n]\setminus\{j,k\}$. Then,
\begin{align*}
\E_{v\sim F}\left[\max_{a\in A}{\sw(a,v)}|X_i(v)=t_i, i\in [n]\right] &\leq
\E_{v\sim F}\left[\max_{a\in A}{\sw(a,v)}|X_i(v)=t'_i, i\in [n]\right].
\end{align*}
\end{lemma}

\begin{proof}
Without loss of generality, we assume that $j=1$ and $k=2$. For each choice of vector $\hat{t} \in \{t,t'\}$, we will abbreviate the event $X_i(v)=\hat{t}_i$ for $i\in [n]$ by $C(\hat{t})$. Define the vector $u = (u_1,u_2,...,u_n) = (t_1-1,t_2,t_3,...,t_n)$ and, for a non-negative integer $r$, let $M(r)$ be the event that there is an alternative $a\in A$ such that $\sum_{i\in[n]}\one\{\pos_{\succ_i}(a)\leq u_i\}\geq r$. We denote by $\overline{M}(r)$ the complement event.
By the properties of the expectation and using the law of total expectation, we have
\begin{align*}
&\E_{v\sim F}\left[\max_{a\in A}\sw(a,v)|C(\hat{t})\right]\\
&= \sum_{r=1}^n\Pr_{v\sim F}\left[\max_{a\in A}\sw(a,v)\geq r|C(\hat{t})\right]\\
&= \sum_{r=1}^n\left(\Pr_{\substack{v\sim F\\P\sim \PP(v)}}\left[M(r)|C(\hat{t})\right] \cdot \Pr_{v\sim F}\left[\max_{a\in A}\sw(a,v)\geq r|M(r),C(\hat{t})\right]\right. \\
&\hspace{36pt}+ \left.\Pr_{\substack{v\sim F\\P\sim \PP(v)}}\left[\overline{M}(r)|C(\hat{t})\right] \cdot \Pr_{v\sim F}\left[\max_{a\in A}\sw(a,v)\geq r|\overline{M}(r),C(\hat{t})\right]\right)\\
&= \sum_{r=1}^n\left(\Pr_{\substack{v\sim F\\P\sim \PP(v)}}\left[M(r)\right]+\Pr_{\substack{v\sim F\\P\sim \PP(v)}}\left[\overline{M}(r)\right] \cdot \Pr_{v\sim F}\left[\max_{a\in A}\sw(a,v)\geq r|\overline{M}(r),C(\hat{t})\right]\right).
\end{align*}
On the last line, the conditions $M(r)$ and $C(\hat{t})$ combined imply that there is indeed an alternative with social welfare at least $r$. We also used the observation (i.e., property 1 above) that the conditions $M(r)$ and $\overline{M}(r)$ only depend on the realization of the agents' (uniformly random) rankings. Now, notice that in the previous equality only the term $\Pr[\max_{a\in A}\sw(a,v)\geq r|\overline{M}(r),C(\hat{t})]$ depends on the choice of $\hat{t}$ among $t,t'$.

In the following, we distinguish between $\hat{t} = t$ (case 1) and $\hat{t} = t'$ (case 2). Let $a_1,a_2$ be the random alternatives appearing in position $t_1$ of agent 1's ranking and position $t_2+1$ of agent 2's ranking. Under the conditions $\overline{M}(r), C(\hat{t})$ combined, only the alternatives $a_1,a_2$ may have social welfare of at least $r$.

\paragraph{Case 1.} $\hat{t} = t = (t_1,t_2,t_3,..., t_n)$. Alternative $a_2$ has social welfare of at most $r-1$ due to $\overline{M}(r), C(t)$ and the fact that $\pos_{\succ_2}(a_2)>t_2$. Hence, only alternative $a_1$ can have social welfare of at least $r$ such that
\begin{align*}
&\Pr_{v\sim F}\left[\max_{a\in A}\sw(a,v)\geq r|\overline{M}(r),C(t)\right]\\
&= \Pr_{v\sim F}\left[\sw(a_1,v)\geq r|\overline{M}(r),C(t)\right]\\
&= \Pr_{\substack{v\sim F\\P\sim \PP(v)}}\left[1+\one\{\pos_{\succ_2}(a_1) \leq t_2\}+\sum_{i=3}^n\one\{\pos_{\succ_i}(a_1)\leq t_i\}\geq r|\overline{M}(r)\right]\\
&= \Pr_{\substack{v\sim F\\P\sim \PP(v)}}\left[\one\{\pos_{\succ_2}(a_1) \leq t_2\}+\sum_{i=3}^n\one\{\pos_{\succ_i}(a_1)\leq t_i\}= r-1\right].
\end{align*}
The last equality follows since, under the condition $\overline{M}(r)$, alternative $a_1$ can appear at positions $t_i$ or above for at most $r-1$ agents among the agents $2,3,...,n$. Applying the law of total probability yields
\begin{align*}
&\Pr_{\substack{v\sim F\\P\sim \PP(v)}}\left[\one\{\pos_{\succ_2}(a_1) \leq t_2\}+\sum_{i=3}^n\one\{\pos_{\succ_i}(a_1)\leq t_i\}= r-1\right]\\
&= \Pr_{\substack{v\sim F\\P\sim \PP(v)}}\left[\sum_{i=3}^n\one\{\pos_{\succ_i}(a_1)\leq t_i\}= r-1\right]\cdot 1\\
&\hspace{24pt}+ \Pr_{\substack{v\sim F\\P\sim \PP(v)}}\left[\sum_{i=3}^n\one\{\pos_{\succ_i}(a_1)\leq t_i\}= r-2\right]\cdot \Pr_{\substack{v\sim F\\P\sim \PP(v)}}\left[\one\{\pos_{\succ_2}(a_1) \leq t_2\}=1\right]\\
&\hspace{24pt}+ \Pr_{\substack{v\sim F\\P\sim \PP(v)}}\left[\sum_{i=3}^n\one\{\pos_{\succ_i}(a_1)\leq t_i\}< r-2\right]\cdot 0\\
&= \Pr_{\substack{v\sim F\\P\sim \PP(v)}}\left[\sum_{i=3}^n\one\{\pos_{\succ_i}(a_1)\leq t_i\}= r-1\right]
 + \Pr_{\substack{v\sim F\\P\sim \PP(v)}}\left[\sum_{i=3}^n\one\{\pos_{\succ_i}(a_1)\leq t_i\}= r-2\right]\cdot \frac{t_2}{m}.
\end{align*}

\paragraph{Case 2.} $\hat{t} = t'= (t_1-1,t_2+1,t_3,..., t_n)$. Similar to the previous case, we obtain
\begin{align*}
&\Pr_{v\sim F}\left[\max_{a\in A}\sw(a,v)\geq r|\overline{M}(r),C(t')\right]\\
&= \Pr_{\substack{v\sim F\\P\sim \PP(v)}}\left[\one\{\pos_{\succ_1}(a_2) \leq t_2\}+\sum_{i=3}^n\one\{\pos_{\succ_i}(a_2)\leq t_i\}= r-1\right]\\
&= \Pr_{\substack{v\sim F\\P\sim \PP(v)}}\left[\sum_{i=3}^n\one\{\pos_{\succ_i}(a_2)\leq t_i\}= r-1\right]
 + \Pr_{\substack{v\sim F\\P\sim \PP(v)}}\left[\sum_{i=3}^n\one\{\pos_{\succ_i}(a_2)\leq t_i\}= r-2\right]\cdot \frac{t_1-1}{m}.
\end{align*}

\vspace{12pt}

We now use the equalities obtained for the two cases to show that, for every $r\in[n]$, the probability $\Pr[\max_{a\in A}\sw(a,v)\geq r|\overline{M}(r),C(\hat{t})]$ is greater for $\hat{t} = t'$ than for $\hat{t} = t$ which proves the lemma.

As highlighted above by property 2, it holds that
$$\Pr_{\substack{v\sim F\\P\sim \PP(v)}}\left[\sum_{i=3}^n\one\{\pos_{\succ_i}(a_1)\leq t_i\}= r-1\right]
= \Pr_{\substack{v\sim F\\P\sim \PP(v)}}\left[\sum_{i=3}^n\one\{\pos_{\succ_i}(a_2)\leq t_i\}= r-1\right],$$
and
$$\Pr_{\substack{v\sim F\\P\sim \PP(v)}}\left[\sum_{i=3}^n\one\{\pos_{\succ_i}(a_1)\leq t_i\}= r-2\right]
= \Pr_{\substack{v\sim F\\P\sim \PP(v)}}\left[\sum_{i=3}^n\one\{\pos_{\succ_i}(a_2)\leq t_i\}= r-2\right]$$
due to the agents' rankings being uniformly random. Finally, by the assumption that $t_1 \geq t_2+2$, we have that $(t_1-1)/m > t_2/m$ such that, for every $r$,
$$\Pr_{v\sim F}\left[\max_{a\in A}\sw(a,v)\geq r|\overline{M}(r),C(t')\right] > \Pr_{v\sim F}\left[\max_{a\in A}\sw(a,v)\geq r|\overline{M}(r),C(t)\right]$$
which concludes the proof.
\end{proof}

We are now ready to prove Lemma~\ref{lem:balanced}. Starting from any vector $t$ with non-negative integer entries and applying Lemma~\ref{lem:t-t-primes} repeatedly, we obtain that
\begin{align*}
&\E_{v\sim F}\left[\max_{a\in A}{\sw(a,v)}|X_i(v)=t_i, i\in [n]\right] \\
&\leq
\E_{v\sim F}\left[\max_{a\in A}{\sw(a,v)}|X_i(v)\in \{\lfloor s/n \rfloor, \lceil s/n \rceil\}, i\in [n], X(v)=s\right]\\
&= \E_{v\sim F}\left[\max_{a\in A}{\sw(a,v)}| \balanced(v,s)\right].
\end{align*}
Then,
\begin{align*}
&\E_{v\sim F}\left[\max_{a\in A}{\sw(a,v)}|X(v)=s\right] \\
&= \sum_{\substack{t=(t_1, ..., t_n):\\\sum_{i\in [n]}{t_i}=s}}{\E_{v\sim F}\left[\max_{a\in A}{\sw(a,v)}|X_i(v)=t_i, i\in [n]\right] \cdot \Pr_{v\sim F}\left[X_i(v)=t_i, i\in [n]|X(v)=s\right]}\\
&\leq \sum_{\substack{t=(t_1, ..., t_n):\\\sum_{i\in [n]}{t_i}=s}}{\E_{v\sim F}\left[\max_{a\in A}{\sw(a,v)}| \balanced(v,s)\right]\cdot \Pr_{v\sim F}\left[X_i(v)=t_i, i\in [n]|X(v)=s\right]}\\
&=\E_{v\sim F}\left[\max_{a\in A}{\sw(a,v)}| \balanced(v,s)\right],
\end{align*}
implying the lemma.
\qed

\section{Proof of Lemma~\ref{lem:j-bound-case-1}}\label{app:sec:lem:j-bound-case-1}
In the proof of Lemmas~\ref{lem:j-bound-case-1} and \ref{lem:t-bound-case-2}, we will use the following simple claim.

\begin{claim}\label{claim:simple-technical-claim}
Let $k$ be a positive integer, $S$ and $S'$ sets of agents with $|S|\leq |S'|$, and $T_i$ a set of $k$ consecutive positions in agent $i\in S$. Then, 
\begin{align*}
\E_{\substack{v\sim F\\P\sim \PP(v)}}\left[\max_{a\in A}{\sum_{i\in S}{\one\{\pos_{\succ_i}(a)\in T_i\}}}\right] \leq 
\E_{\substack{v\sim F\\P\sim \PP(v)}}\left[\max_{a\in A}{\sum_{i\in S'}{\one\{\pos_{\succ_i}(a)\leq k\}}}\right] 
\end{align*}
\end{claim}

To see why the claim holds, observe that the expectations are defined over uniformly random profiles. Hence, the probability that an alternative appears in any position of any agents' ranking is equal to $1/m$, i.e., independent from both the agent and the position. 

Returning to the proof of Lemma~\ref{lem:j-bound-case-1}, 
define the sets of agents $N_1=\{1, ..., \lfloor n/2\rfloor\}$, $N_2=\{1, ..., \lfloor n/2\rfloor+1, ..., 2\lfloor n/2 \rfloor\}$, and $N_3=N\setminus (N_1\cup N_2)$. We have \begin{align*}
&\E_{v\sim F}\left[\max_{a\in A}{\sw(a,v)}|\balanced(v,jn\tau)\right]\\
&= \E_{\substack{v\sim F\\P\sim \PP(v)}}\left[\max_{a\in A}{\sum_{i\in N}{\one\{\pos_{\succ_i}(a)\leq j\tau\}}}\right]\\
&= \E_{\substack{v\sim F\\P\sim \PP(v)}}\left[\max_{a\in A}{\sum_{\ell\in [3]}{\sum_{i\in N_\ell}}\sum_{k=1}^j{\one\{(k-1)\tau <\pos_{\succ_i}(a)\leq k\tau\}}}\right]\\
&\leq \sum_{\ell\in [3]}\sum_{k=1}^j\E_{\substack{v\sim F\\P\sim \PP(v)}}\left[\max_{a\in A}{{\sum_{i\in N_\ell}}{\one\{(k-1)\tau <\pos_{\succ_i}(a)\leq k\tau\}}}\right]\\
&\leq \sum_{\ell\in [3]}\sum_{k=1}^j\E_{\substack{v\sim F\\P\sim \PP(v)}}\left[\max_{a\in A}{{\sum_{i\in N_1}}{\one\{\pos_{\succ_i}(a)\leq \tau\}}}\right]\\
&= 3j\cdot \E_{v\sim F}\left[\max_{a\in A}{\sw(a,v)}|\boxv(v)\right]=3j\cdot B.
\end{align*}
The first equality is due to the fact that the condition $\balanced(v,jn\tau)$ requires that the contribution to the social welfare of each alternative comes from positions from $1$ to $j\tau$ only. The first inequality uses a simple property of the maximum function and linearity of expectation, while the second one follows from Claim~\ref{claim:simple-technical-claim}. The last equality follows since the condition $\boxv(v)$ requires that the contribution to the social welfare of each alternative comes from positions from position $1$ to $\tau$ of exactly $\lfloor n/2 \rfloor$ agents. The last equality uses the definition of $B$.
\qed

\section{Proof of Lemma~\ref{lem:t-bound-case-2}}\label{app:sec:lem:t-bound-case-2}
Let $N_1=[t]$. Define the condition $\pbalanced(v,t,s)$ to be $Z(v,1)=t$, $\lfloor s/t\rfloor \leq X_i(v)\leq \lceil s/t \rceil$ for $i\in [N_1]$ and $X_i(v)=0$ for $i\in N\setminus N_1$. Using Lemma~\ref{lem:t-t-primes} in a similar way we used it to prove Lemma~\ref{lem:balanced}, we can show that
\begin{align*}
\E_{v\sim F}\left[\max_{a\in A}{\sw(a,v)}|Z(v,1)=t, X(v)=j\right] &\leq \E_{v\sim F}\left[\max_{a\in A}{\sw(a,v)}|\pbalanced(v,t,j)\right].
\end{align*}
So, we have 
\begin{align*}
&\E_{v\sim F}\left[\max_{a\in A}{\sw(a,v)}|Z(v,1)=t, X(v)=j\right]\\
&\leq \E_{v\sim F}\left[\max_{a\in A}{\sw(a,v)}|\pbalanced(v,t,j)\right]\\
&\leq \E_{v\sim F}\left[\max_{a\in A}{\sw(a,v)}|\pbalanced(v,t,\lceil j/t\rceil\cdot t)\right]\\
&= \E_{\substack{v\sim F\\P\sim \PP(v)}}\left[\max_{a\in A}{\sum_{i\in N_1}{\one\{\pos_{\succ_i}(a)\leq \lceil j/t\rceil\}}}\right]\\
&=\E_{\substack{v\sim F\\P\sim \PP(v)}}\left[\max_{a\in A}{\sum_{i\in N_1}\sum_{k=1}^{\lceil j/t\rceil}{\one\{\pos_{\succ_i}(a)=k\}}}\right]\\
&\leq \sum_{k=1}^{\lceil j/t\rceil}\E_{\substack{v\sim F\\P\sim \PP(v)}}\left[\max_{a\in A}{\sum_{i\in N_1}{\one\{\pos_{\succ_i}(a)=k\}}}\right]\\
&\leq \lceil j/t \rceil \cdot \E_{\substack{v\sim F\\P\sim \PP(v)}}\left[\max_{a\in A}{\sum_{i\in N_1}{\one\{\pos_{\succ_i}(a)=1\}}}\right]\\
&= \lceil j/t \rceil \cdot\E_{\substack{v\sim F\\P\sim \PP(v)}}\left[\max_{a\in A} \impliedsw(a,P,v,1)|Z(v,1) = t\right].
\end{align*}
The second inequality holds since the condition $\pbalanced(v,t,\lceil j/t\rceil \cdot t)$ implies the condition $\pbalanced(v,t,j)$ and the quantity $\E[\max_{a\in A}{\sw(a,v)}|\pbalanced(v,t,j)]$ is non-decreasing in $j$. The first equality is due to the fact that condition $\pbalanced(v,jn\tau)$ requires that the contribution to the social welfare of each alternative comes from positions $1$ to $\lceil j/t\rceil$ of $t$ agents only. Next, the fourth inequality follows from a simple property of the maximum function and linearity of expectation, while the fifth inequality follows from Claim~\ref{claim:simple-technical-claim}. Finally, for the last inequality, observe that under the condition $Z(v,1)=t$, the implied social welfare of each alternative comes from the top position of a set of $t$ agents.
\qed

\section{Two comments regarding mechanism \med\ and its analysis}\label{app:sec:comments}
We devote this section to discussing two issues related to mechanism \med. First, notice that \med\ queries the value of the top-ranked alternative in each agent for $p<2/m$ and an alternative at a lower position in the agents' rankings otherwise. Let us consider the simpler variant of \med\ which always queries the value of each top-ranked alternative and returns the one of maximum implied social welfare. Our Lemma~\ref{lem:log-over-loglog} below shows that this variant of \med\ has super-constant distortion when ties in implied social welfare are resolved arbitrarily, that is, ignoring the content of the profile below the top position in each ranking. 

Our proof uses a profile with $m$ alternatives, $n=m^{1/3}$ agents, and $p=m^{-1/3}$. In this way, there are, on average, $m^{2/3}$ alternatives that have a value of $1$ in each agent, but the mechanism recovers very little information about which alternatives do get a value of $1$. In particular, our choice of parameters $n$, $m$, and $p$ guarantees that all top-ranked alternatives are different with high probability. Then, the alternative selected among them by the mechanism has only constant expected social welfare. In contrast, the parameters are such that the expected maximum social welfare is $\Omega\left(\frac{\log{m}}{\log\log{m}}\right)$.

\begin{lemma}\label{lem:log-over-loglog}
Let $\mechanism$ be the mechanism that queries the value of the top-ranked alternative in every agent and returns an alternative that maximizes the implied social welfare (breaking ties arbitrarily). It holds that $\adist(\mechanism,\binarydist)\in\Omega\left(\frac{\log{m}}{\log\log{m}}\right)$.
\end{lemma}

\begin{proof}
Let $n\geq 2$, $m = n^3$, and consider the binary distribution $F_p$ with $p = 1/n = m^{-1/3}$. We first lower-bound the probability that all alternatives appearing in the first position of any agent's ranking are different. For a set of binary valuations $v$, we refer to this condition as $\distinct(v)$. Recall that, in an impartial culture electorate, the top-ranked alternative of each agent is, in effect, uniformly random and independent from the top-ranked alternatives of the remaining agents. We thus have that
\begin{equation}\label{eqn:all_top_ranked_distinct}
\Pr_{v\sim F_p}[\distinct(v)]
= \prod_{i=0}^{n-1}\frac{m-i}{m}
= \prod_{i=0}^{m^{1/3}-1}\left(1-\frac{i}{m}\right)
\geq \left(1-m^{-2/3}\right)^{m^{1/3}}
\geq 1 - m^{-1/3}.
\end{equation}
Here, we used the fact that $(1+x)^r \geq 1+rx$ for $x\geq -1, r\in\R\setminus(0,1)$.

Now, notice that the probability that an alternative gets a value of $1$ from an agent, given that it is not top-ranked by this agent, is at most $p$. Then, the expected social welfare of the alternative returned by mechanism $\mechanism$ under the condition $\distinct(v)$ is at most $1+p(n-1)\leq 2$. Thus,
\begin{align}\nonumber
\E_{\substack{v\sim F_p\\ P \sim \PP(v)}}\left[\sw(\mechanism(P,v), v)\right] & \leq \Pr_{v\sim F_p}\left[\distinct(v)\right]\cdot \E_{\substack{v\sim F_p\\P\sim \PP(v)}}\left[\sw(\mechanism(P,v),v)|\distinct(v)\right]\\\nonumber
&\quad \quad +(1-\Pr_{v\sim F_p}\left[\distinct(v)\right])\cdot n\\\label{eqn:exp_social_welfare_top_ranked_query_algo}
&\leq 2 + m^{-1/3}\cdot n=3.
\end{align}
The first inequality follows from the law of total expectation together with the observation that the expected social welfare of any alternative is trivially upper-bounded by $n$ under binary valuations. For the second inequality, we use inequality~(\ref{eqn:all_top_ranked_distinct}). 

We complete the proof by showing that the expected maximum social welfare is at least $\Omega\left(\frac{\log{m}}{\log\log{m}}\right)$. We do so by proving that the probability that the maximum social welfare is at least $\left\lfloor \frac{\log{m}}{\log\log{m}}\right\rfloor$ is lower-bounded by a constant. First, notice that
\begin{align*}
\Pr_{v\sim F_p}[\sw(a,v)\geq k]
&\geq \Pr_{v\sim F_p}[\sw(a,v) = k]
 = \binom{n}{k}\cdot p^k\cdot \left(1- p\right)^{n-k}\\
& \geq \left(\frac{m^{1/3}}{k}\right)^k \cdot m^{-k/3}\cdot \left(1- m^{-1/3}\right)^{m^{1/3}} = \frac{\left(1- m^{-1/3}\right)^{m^{1/3}}}{k^k},
\end{align*}
where the second inequality follows since $\binom{n}{k}\geq \left(\frac{n}{k}\right)^k$. We now observe that $(1- m^{-1/3})^{m^{1/3}}$ is strictly increasing in $m$ approaching $e^{-1}$ from below. For $n\geq 2$, we have that $m\geq 8$ and, thus, $(1- m^{-1/3})^{m^{1/3}} \geq\frac14$. This yields
$\Pr_{v\sim F_p}[\sw(a,v)\geq k]\geq\frac{1}{4k^k}$
and, hence
\begin{align}\label{eqn:pr_no_alternative_with_sw_k}
\Pr_{v\sim F_p}\left[\max_{a\in A}{\sw(a,v)}\geq k\right] &\geq 1- \left(1-\frac{1}{4k^k}\right)^m\geq 1-\exp\left(-\frac{m}{4k^k}\right),
\end{align}
where the last inequality is due to the fact that $e^x\geq 1+x$ for any real $x$. By selecting $k=\left\lfloor \frac{\log{m}}{\log\log{m}}\right\rfloor$, we have $k^k\leq m$ and 
\begin{align*}
    \Pr_{v\sim F_p}\left[\max_{a\in A}{\sw(a,v)}\geq \left\lfloor \frac{\log{m}}{\log\log{m}}\right\rfloor\right] &\geq 1-e^{-1/4}> 0.22,
\end{align*}
as desired.
\end{proof}

One may still wonder whether concentration inequalities like Chernoff bounds could replace (parts of) our analysis of mechanism \med. Intuitively, if the expected social welfare of each alternative is high (e.g., $n p\in \Omega(\log{m})$), then the social welfare of all alternatives will be sharply concentrated around this expectation, and the expected maximum and expected minimum social welfare will only be a constant factor apart. This would imply constant average distortion for {\em all} mechanisms, including those that make no queries at all. We include a formal proof of this fact as Lemma~\ref{lemma:const_av_dist_binary} below. Unfortunately, the assumption that $n p \in \Omega(\log{m})$ required by the lemma does not subsume any of the three cases in our analysis of mechanism \med\ in Section~\ref{sec:1-query}. Furthermore, we do not see how to extend the use of concentration inequalities to a broader range of parameters (e.g., satisfying $n p \in o(\log{m})$) where we have proven that queries are necessary.

\begin{lemma}\label{lemma:const_av_dist_binary}
Any voting rule has average distortion at most $13$ in impartial culture electorates with $n$ agents and $m$ alternatives, and underlying values drawn from a binary distribution $F_p$ such that $n\cdot p \geq 8\ln{(2m)}$.
\end{lemma}

\begin{proof}
We will show that returning any alternative (including the one of minimum social welfare) yields an average distortion of at most $13$. We will use the upper tail Chernoff bound (Lemma~\ref{lem:chernoff}) as well as its next lower tail version.
\begin{lemma}[Chernoff bound, lower tail]\label{lem:chernoff-lower}
For every binomial random variable $Q$ and any $\delta\in [0,1]$, we have
\begin{align*}
\Pr[Q\leq (1-\delta)\E[Q]] \leq \exp\left(-\frac{\delta^2\E[Q]}{2}\right).
\end{align*}
\end{lemma}

For valuations $v$ drawn according to $F_p$ and an alternative $a\in A$, let $Q_i(v)$ be a random variable that indicates whether agent $i$ has value $1$ (then, $Q_i(v) = 1$) or $0$ (then, $Q_i(v) = 0$) for this alternative. Define $Q(v) = \sum_{i\in N} Q_i(v)$ and $\mu = \E_{v\sim F_p}[Q(v)]$. Thereby, $Q$ is the social welfare of alternative $a$, and $\mu$ is its expected value. By our assumption, it holds that $\mu = n\cdot p \geq 8\ln{(2m)}$.

Notice that for $\delta\geq 2$, we have $\frac{\delta^2}{2+\delta}\geq \frac{\delta}{2}$ and Lemma~\ref{lem:chernoff} now implies that
\begin{align}\label{eq:simpler-chernoff-ineq}
    \Pr_{v\sim F_p}[Q(v)\geq (1+\delta)\mu]&\leq \exp\left(-\frac{\delta \mu}{2}\right).
\end{align}
For any $t\geq 3\mu$, we use the union bound and apply inequality (\ref{eq:simpler-chernoff-ineq}) with $\delta=\frac{t}{\mu}-1\geq 2$ to obtain 
\begin{align}\label{eq:simpler-chernoff-ineq-2}
\Pr_{v\sim F_p}\left[\max_{a\in A} \sw(a,v)\geq t\right]
&\leq m\cdot \Pr_{v\sim F_p}[Q(v)\geq t] \leq  m\cdot\exp\left(-\frac{t - \mu}{2}\right).
\end{align}
By the definition of the expectation and using inequality (\ref{eq:simpler-chernoff-ineq-2}), it holds that
\begin{align}\nonumber
\E_{v\sim F_p}\left[\max_{a\in A}\sw(a,v)\geq t\right]
&= \int_0^\infty \Pr_{v\sim F_p}\left[\max_{a\in A}\sw(a,v)\geq t\right] dt\\\nonumber
&\leq \int_0^{3\mu} dt + \int_{3\mu}^\infty \Pr_{v\sim F_p}\left[\max_{a\in A}\sw(a,v)\geq t\right] dt\\\nonumber
&\leq 3\mu + m\int_{3\mu}^\infty \exp\left(-\frac{t-\mu}{2}\right) dt\\\label{eq:simpler-chernoff-ineq-3}
&= 3\mu+2m\exp(-\mu)\leq 3\mu+(2m)^{-7},
\end{align}
where the last inequality follows since $\mu\geq 8\ln{(2m)}$. 

Now, using again the union bound and the lower tail Chernoff bound (Lemma~\ref{lem:chernoff-lower}), we get
\begin{align*}
    \Pr_{v\sim F_p}\left[\min_{a\in A}{\sw(a,v)}\leq \frac{\mu}{2}\right] &\leq m\cdot \Pr_{v\sim F_p}\left[Q(v)\leq \frac{\mu}{2}\right] \leq m\cdot \exp\left(-\frac{\mu}{8}\right) \leq \frac12.
\end{align*}
Thus,
\begin{align}\label{eq:simpler-chernoff-ineq-4}
\E_{v\sim F_p}\left[\min_{a\in A}{\sw(a,v)}\right] &\geq \frac{\mu}{2}\cdot \Pr_{v\sim F_p}\left[\min_{a\in A}{\sw(a,v)}\geq \frac{\mu}{2}\right]> \frac{\mu}{4}.
\end{align}
By inequalities (\ref{eq:simpler-chernoff-ineq-3}) and (\ref{eq:simpler-chernoff-ineq-4}), we obtain that the ratio between $\E_{v\sim F_p}\left[\max_{a\in A}{\sw(a,v)}\right]$ and $\E_{v\sim F_p}\left[\min_{a\in A}{\sw(a,v)}\right]$ and, consequently, the distortion of any mechanism is at most $13$, as desired.
\end{proof}

\end{document}